\documentclass[runningheads]{llncs}
\usepackage{amsmath,amssymb}
\usepackage{lineno}
\usepackage{bm,upgreek}
\usepackage[noend]{algpseudocode}
\usepackage{algorithm,algorithmicx}
\usepackage{hyperref}
\usepackage{pgf,tikz,pgfplots,subfigure}
\usetikzlibrary{shapes,arrows,automata}
\usepackage[braket,qm]{qcircuit}
\usepackage{dsfont}
\usepackage{multicol}
\usepackage{graphics}
\usepackage[utf8]{inputenc}

\usepackage{float}
\usepackage{booktabs}
\usepackage{multirow}
\usepackage{threeparttable}
\usepackage{tabularx}
\usetikzlibrary{positioning, calc, decorations.pathreplacing}
\usepackage{array} 

\usetikzlibrary{calc, intersections, patterns}
\usetikzlibrary{spy}

\newcommand{\QC}{\mathfrak{Q}}
\newcommand{\h}{\mathbb{H}}
\newcommand{\den}{\mathcal{D}}
\newcommand{\I}{\mathcal{I}}
\newcommand{\BB}{\mathbb{B}}
\newcommand{\II}{\mathbb{I}}

\newcommand{\M}{\mathds{M}}
\newcommand{\HH}{\mathbf{H}}
\newcommand{\LL}{\mathbf{L}}
\newcommand{\LLL}{\mathcal{L}}
\newcommand{\PP}{\mathbf{P}}
\newcommand{\T}{\mathrm{T}}
\newcommand{\dd}{\mathrm{d}}
\newcommand{\id}{\mathbf{I}}
\newcommand{\J}{\mathcal{J}}

\newcommand{\vl}{\mathrm{V2L}}
\newcommand{\lv}{\mathrm{L2V}}
\newcommand{\e}{\mathrm{e}}
\newcommand{\tr}{\mathrm{tr}}

\newcommand{\x}{\mathbf{x}}
\newcommand{\ntl}{\mathrm{U}\,}
\newcommand{\mnt}{\mathrm{mnt}}
\newcommand{\true}{\texttt{true}}

\newcommand{\mybrace}[4]{ \draw[#4, decoration={brace, amplitude=4}, decorate] ([yshift=6mm]#1)--([yshift=6mm]#2) node[midway, above=1.5mm]{#3}; }
\newcommand{\mybracesmall}[4]{ \draw[#4, decoration={brace, amplitude=4}, decorate] ([yshift=6mm]#1)--([yshift=6mm]#2) node[midway, above=1.5mm]{#3}; }

\newcommand{\mybracesmallbelow}[4]{ \draw[#4, decoration={brace, mirror, amplitude=4}, decorate] ([yshift=-6mm]#1)--([yshift=-6mm]#2) node[midway, below=1.5mm]{#3}; }

\pgfplotsset{compat=1.18,every axis/.append style={
line width=0.8pt,
mark size=1pt,
tick style={line width=0.6pt}}}
\definecolor{pink1}{RGB}{236,189,180}
\definecolor{blue1}{RGB}{99,157,217}

\definecolor{red1}{RGB}{200,74,72}
\definecolor{blue2}{RGB}{43,123,175}
\definecolor{green1}{RGB}{157,204,94}

\usepackage{soul}

\begin{document}
\title{A Sample-Driven Solving Procedure for the Repeated Reachability of Quantum CTMCs}
%
%
\author{Hui Jiang$^a$ \and Jianling Fu$^b$ \and Ming Xu$^a$ \and Yuxin Deng$^a$ \and Zhi-Bin Li$^b$}
\titlerunning{Repeated Reachability in Quantum CTMCs}
\authorrunning{Jiang, H., et al.}
%
\institute{$^a$Shanghai Key Lab of Trustworthy Computing,
East China Normal University, China \\
$^b$Shanghai Institute for AI Education, East China Normal University, China
}

\maketitle              
\begin{abstract}
	Reachability analysis plays a central role in system design and verification.
	The reachability problem, denoted $\Diamond^\J\,\Phi$,
	asks whether the system will meet the property $\Phi$ after some time in a given time interval $\J$.
	Recently, it has been considered on a novel kind of real-time systems
	--- quantum continuous-time Markov chains (QCTMCs),
	and embedded into the model-checking algorithm.
	In this paper, we further study the repeated reachability problem in QCTMCs,
	denoted $\Box^\I\,\Diamond^\J\,\Phi$,
	which concerns whether the system starting from each \emph{absolute} time in $\I$
	will meet the property $\Phi$ after some coming \emph{relative} time in $\J$.
	First of all, we reduce it to
	the real root isolation of a class of real-valued functions (exponential polynomials),
    whose solvability is conditional to Schanuel's conjecture being true.
	To speed up the procedure, we employ the strategy of sampling.
	The original problem is shown to be equivalent to
	the existence of a finite collection of satisfying samples.
	We then present a sample-driven procedure,
	which can effectively refine the sample space after each time of sampling,
	no matter whether the sample itself is successful or conflicting.
	The improvement on efficiency is validated by randomly generated instances.
	Hence the proposed method would be promising to attack the repeated reachability problems
	together with checking other $\omega$-regular properties in a wide scope of real-time systems.

\end{abstract}

\section{Introduction}
As one of the most important stochastic processes in classical world,
the model of Markov chain has been extensively studied since the early 20th century.
To embrace an increased degree of realism to describe random events in practice,
continuous-time random variables are incorporated into this model
which is called the continuous-time Markov chain (CTMC).
Meanwhile,
the rapid advancement of quantum computing over the past few decades
has led to the flourishing of models in quantum world,
e.\,g.\@, quantum automaton~\cite{KoW97},
quantum discrete-time Markov chain (QDTMC)~\cite{FYY13a}
and quantum discrete-time Markov decision process (QDTMDP)~\cite{YiY18}.
For widely-applied CTMC,
researchers have shown great interest in its quantum analogues
--- quantum continuous-time Markov chain (QCTMC),
when it collides with quantum mechanics.

The motion planning problem can be interpreted on Markov models
to accomplish complex tasks within rigorous temporal constraints.
One way of expressing such tasks is to use various temporal logics,
such as computation tree logic (CTL)~\cite{AnM95}
and linear temporal logic (LTL)~\cite{LoK04}.
In the approach to checking LTL formulas,
the execution of the system is thought of a sequence of states or events.
This representation abstracts from the precise timing of observations,
retaining only the order on states or events.
Metric temporal logic (MTL)~\cite{OuW05} and signal temporal logic (STL)~\cite{MaN04}
are able to express the specification of systems in quantitative timing.
MTL is a temporal logic specified on a discrete-time specification,
while STL is a variant of MTL tailored to specify the properties of continuous-time signals.
In this paper, we consider real-time systems;
therefore STL is preferable.
The core of STL lies in reachability-like properties.

On the other hand, quantum mechanics allows people to understand an astonishing range of phenomena in the world,
such as superposition, mixture and entanglement.
These phenomena have been utilized to design many kinds of quantum systems and protocols~\cite{BeB84,QDD20},
and becoming more and more crucial in motion planning.
It is necessary to investigate the dynamics of such novel quantum systems. 
In this paper, we will study the repeated reachability problem in finite horizon on QCTMCs,
denoted the temporal formula $\Box^\I\,\Diamond^\J\,\Phi$,
which concerns whether the system at each \emph{absolute} time in $\I$
will meet the property $\Phi$ after some coming \emph{relative} time in $\J$.
Our contributions are three-fold:
\begin{enumerate}
    \item We show the decidable result by
		a reduction to the real root isolation of a class of real-valued functions (exponential polynomials),
		following the state-of-the-art number-theoretic tool --- Schanuel's conjecture~\cite{Ax71}.
    \item To speed up the procedure, we employ the strategy of sampling.
		The original problem is shown to be equivalent to
		the existence of a finite collection of satisfying samples.
    \item We then present a sample-driven procedure,
		which can effectively refine the sample space after each time of sampling,
		no matter whether the sample itself is successful or conflicting.
\end{enumerate}
The improvement on efficiency is validated by randomly generated instances.
Hence the proposed method would be promising in attacking the repeated reachability problems
together with checking other $\omega$-regular properties in a wide scope of real-time systems.

\subsection{Related Work on Reachability Problems}
Reachability analysis, embedded in the model-checking algorithms,
is essential to verify both classical systems and quantum ones.
We discuss them as follows.

\paragraph{Verification on classical MCs}
Early in 1985, Moshe Y. Vardi initiated the verification of probabilistic concurrent finite-state programs
(a.\,k.\,a.\@ discrete-time Markov chains, DTMCs)~\cite{Var85},
in which the reachability and the repeated reachability problems are expressed respectively
by LTL formulas $\Diamond\,\Phi$ and $\Box\,\Diamond\,\Phi$ for some static observable $\Phi$.
The general problem of probabilistic model checking
with respect to the $\omega$-regular specification was considered in~\cite{BRV04}.
LTL focuses on the properties of linear processes;
to specify the properties on branching processes, people resort to CTL~\cite{CES86}.
Hansson and Jonsson introduced probabilistic CTL (PCTL) by adding the probability-quantifier~\cite{HaB94} on DTMCs.
Almost sure repeated reachability is PCTL-definable and computable.

\paragraph{Verification on classical CTMCs}
The seminal work on verifying CTMCs is put forward by
Aziz \textit{et~al.}’s and Baier \textit{et~al.}’s~\cite{ASS+96,BKH99}.
They introduced continuous stochastic logic (CSL) interpreted on CTMCs,
which is the extension of discrete-time stochastic systems. 
In \cite{ASS+96}, probabilities are required to be rational numbers,
and the decidability of model checking for CSL is accomplished using number-theoretic results.
An approximate model checking algorithm for a reduced version of CSL was provided in~\cite{BKH99},
which restricted path formulas
from multiphase until formulas $\Phi_1 \ntl^{\I_1}\Phi_2 \cdots \ntl^{\I_k}\Phi_{k+1}$
for some integer $k\ge 1$
to binary until ones $\Phi_1 \ntl^{\I}\Phi_2$,
a kind of constrained reachability.
Under this logic, they applied efficient numerical techniques --- uniformization~\cite{Ste94} ---
for transient analysis~\cite{BHH+03}.
The approximate algorithms have been extended for multiphase until formulas using stratification~\cite{ZJN+11}. 
Xu \textit{et~al.} considered the multiphase until formulas over the CTMC with rewards~\cite{XZJ+16};
some positive progress was established by number-theoretic and algebraic methods. 
Recently, the gap between the exact and the approximate methods was bridged in~\cite{FeZ17}.
The inner temporal operator $\Diamond\,\Phi$ of $\Box\,\Diamond\,\Phi$ are qualitative,
and not interpreted well by branching processes, e.\,g.\@ CSL.
An alternative is considering linear processes as in~\cite{GuY22} and the current paper.

\paragraph{Verification on quantum MCs}
In 2013, Feng \textit{et~al.} initiated the verification of QDTMCs~\cite{FYY13a},
in which Markov chains are equipped with the quantum operations as transitions.
Under the model, the authors considered the reachability probability~\cite{YFY+13},
the repeated reachability probability~\cite{FHT+17},
and the model-checking of a quantum analogy of CTL~\cite{FYY13a} with extension~\cite{XFM+22}. 
An exact method was developed in~\cite{XHF21}
to solve the constrained reachability problem for QDTMCs.
Recently, Xu \textit{et~al.} investigated the novel model of QCTMCs,
on which the decidability of the STL formula was established in~\cite{XMGY21}
by a reduction to real root isolation of exponential polynomials.
STL concerns the reachability of linear processes,
and CSL concerns the reachability of branching processes,
whose decidability was settled in~\cite{MXG+22}.
Whereas, as far as we know,
the repeated reachability has not been considered on QCTMCs.

Finally those seminal work on verifying various Markov models are summarized in Table~\ref{table:works},
in which the time complexity is specified in the size of the input model.
When we deal with continuous-time Markov models,
it is essential to sufficiently approach the Euler constant $\mathrm{e}$ for establishing the decidability,
which is left as a common oracle of the existing methods.
\begin{table}[ht]
\caption{Seminal work on verifying various Markov models}\label{table:works}
\centering
\begin{tabular}{p{3cm}<{\centering}ccp{2cm}<{\centering}}
\toprule
Markov models & reachability & repeated reachability \\
\midrule
DTMC & polynomial time~\cite{Var85} & polynomial time~\cite{Var85} \\
CTMC & decidable~\cite{ASS+96} &  not available \\
QDTMC & polynomial time~\cite{YFY+13} & polynomial time~\cite{FHT+17} \\
QCTMC & decidable~\cite{XMGY21} & the present paper \\
\bottomrule
\end{tabular}
\end{table}

\paragraph*{Organization}
The rest of this paper is organized as follows.
Section~\ref{S2} reviews basic notions and notations from quantum computing,
together with the model of QCTMC and STL.
We state our main result on the repeated reachability problem in Section~\ref{S4},
and give the detailed construction respectively in Sections~\ref{S5} \&~\ref{S6}.
Section~\ref{S7} delivers the experimentation.
The paper is concluded in Section~\ref{S8}.

\section{Preliminaries}\label{S2}
\subsection{Basic Notions and Notations}
Let $\h$ be a finite-dimensional Hilbert space
that is a complete vector space over complex numbers $\mathbb{C}$ equipped with an inner product operation
throughout this paper,
and $d$ its dimension.
We recall the standard Dirac notations from quantum computing.
Interested readers can refer to~\cite{NiC00} for more details.
\begin{itemize}
	\item $\ket{\psi}$ denotes a unit column vector in $\h$ labelled with $\psi$;
	\item $\bra{\psi}:=\ket{\psi}^\dag$ is the Hermitian adjoint
	(conjugate transpose entrywise) of $\ket{\psi}$;
	\item $\ip{\psi_1}{\psi_2}:=\bra{\psi_1} \ket{\psi_2}$
	is the inner product of $\ket{\psi_1}$ and $\ket{\psi_2}$;
	\item $\op{\psi_1}{\psi_2}:=\ket{\psi_1} \otimes \bra{\psi_2}$ is the outer product
	where $\otimes$ denotes tensor product;
	\item $\ket{\psi,\psi'}$ is a shorthand of
	the tensor product $\ket{\psi}\ket{\psi'}=\ket{\psi}\otimes\ket{\psi'}$.
\end{itemize}

\paragraph*{Quantum state}
Let $\gamma$ be a linear operator on $\h$.
It is \emph{Hermitian} if $\gamma=\gamma^\dag$;
it is \emph{positive} if $\bra{\psi}\gamma\ket{\psi} \ge 0$ holds for any vector $\ket{\psi}\in\h$.
A \emph{projector} $\PP$ is a positive operator of
the form $\sum_{i=1}^m \op{\psi_i}{\psi_i}$ with $m\le d$,
where $\ket{\psi_i}$ ($i=1,2,\dots,m$) are orthonormal.
It implies that all eigenvalues of $\PP$ are in $\{0,1\}$.
The \emph{trace} of a linear operator $\gamma$ is defined as
$\tr(\gamma):=\sum_{i=1}^d \bra{\psi_i} \gamma \ket{\psi_i}$
for any orthonormal basis $\{\ket{\psi_i}:i=1,2,\dots,d\}$ of $\h$.
A \emph{density} operator $\rho$ is a positive operator with unit trace.
Let $\den$ be the set of density operators.
For a density operator $\rho$,
we have the spectral decomposition
$\rho=\sum_{i=1}^m \lambda_i \op{\lambda_i}{\lambda_i}$
where $\lambda_i$ ($i=1,2,\dots,m$) are positive eigenvalues.
We call such eigenvectors $\ket{\lambda_i}$ \emph{eigenstates} of $\rho$ explained below.
The density operators are usually used to describe quantum states.
We assume the quantum system is in that decomposition,
meaning that it is in state $\ket{\lambda_i}$,
which is an event occurring with probability $\lambda_i$.
When $m=1$, we know that the system is surely in state $\ket{\lambda_1}$ (with probability one),
which is a so-called \emph{pure} state;
otherwise the state is \emph{mixed}.
Both the vector notation $\ket{\lambda_i}$ and the outer product notation $\op{\lambda_i}{\lambda_i}$
could be used to denote pure states;
it is preferable to use the matrix notation $\rho$ to denote mixed states.

We will review the quantum operation on quantum states,
embedded in the following model description of quantum continuous-time Markov chain.

\subsection{Quantum CTMC}
We now introduce the model of quantum CTMCs.
\begin{definition}\label{def:QCTMC}
	A quantum continuous-time Markov chain (QCTMC for short) $\QC$
	is given by a pair $(\h,\LLL)$, in which
	\begin{itemize}
		\item $\h$ is the Hilbert space,
		\item $\LLL$ is the transition generator function given by
		a Hermitian operator $\HH$
		and a finite set of linear operators $\LL_j$ on $\h$.
	\end{itemize}
	Usually, a density operator $\rho(0) \in \den$
	is appointed as the initial state of $\QC$.
\end{definition}

In the model,
the transition generator function $\LLL$ gives rise to a \emph{universal} way
to describe the continuous-time dynamics of the QCTMC,
following 
the Lindblad's master equation~\cite{Lin76,GKS76}
\begin{equation}\label{eq:Lindblad}
\begin{aligned}
    \rho'(t) & =\LLL(\rho(t)) \\
	& =-\imath\HH\rho(t)+\imath\rho(t)\HH
	+\sum_{j=1}^m \left(\LL_j\rho(t) \LL_j^\dag
	-\tfrac{1}{2}\LL_j^\dag\LL_j\rho(t)-\tfrac{1}{2}\rho(t)\LL_j^\dag\LL_j\right).
\end{aligned}
\end{equation}
The above equation is general enough to describe the dynamics of open systems.
If we are concerned with a \emph{closed} system
(an ideal system that does not suffer from any unwanted interaction from outside environment),
state transitions can be characterized by the Schr\"odinger equation:
\begin{subequations}
\begin{equation}\label{eq:schrodinger}
	\frac{\dd\,\ket{\psi(t)}}{\dd\, t} = -\imath \HH\ket{\psi(t)}.
\end{equation}
where $\ket{\psi(t)}$ is the state of the system at time $t$,
and $\HH$ is a Hermitian operator called \emph{Hamiltonian}.
We can reformulate it with matrix notation as
\begin{equation}
	\rho'(t) =-\imath\HH\rho(t)+\imath\rho(t)\HH,
\end{equation}
\end{subequations}
where $\rho=\op{\psi}{\psi}$.
An \emph{open} system interacts with environment.
Composed with the environment, the large system is closed;
by tracing out the environment of the large system,
it is characterized by Eq.~\eqref{eq:Lindblad} where $\LL_j$ are a few linear operators.
For the consideration of computability,
the entries of $\HH$ and $\LL_j$ are supposed to be \emph{algebraic} numbers
that are roots of the polynomials with rational coefficients.
For instance, the number $\tfrac{1}{3}-\imath\sqrt{2}$ is algebraic
since it is a root of $x^2-\tfrac{2}{3}x+\tfrac{19}{9}$.

\begin{example}\label{ex1}
    We consider a sample QCTMC $\QC_1=(\h,\LLL)$,
    in which $\h$ is a Hilbert space over two qubits,
	i.\,e.\@ a $4$-dimensional vector space,
    and the transition function $\LLL$ is given by
	\begin{itemize}
	\item the Hermitian operator $\HH=X\otimes X$, and
	\item the set of linear operators $\{\LL_1, \LL_2\}$ with $\LL_1 =X \otimes H$
    and $\LL_2=H\otimes X$.
	\end{itemize}
    Here, $H$ is the Hadamard operator $\op{+}{0}+\op{-}{1}$
	with $\ket{\pm}=(\ket{0}\pm\ket{1})/\sqrt{2}$,
    and $X$ is the Pauli-X operator $\op{1}{0}+\op{0}{1}$.
    Once the initial state $\rho(0)$ is fixed,
    the dynamics of $\QC_1$ is entirely determined by Lindblad's master equation
\[
    \rho'(t) = -\imath\HH\rho(t)+\imath\rho(t)\HH
	+\sum_{j=1}^2 \left(\LL_j\rho(t) \LL_j^\dag
	-\tfrac{1}{2}\LL_j^\dag\LL_j\rho(t)-\tfrac{1}{2}\rho(t)\LL_j^\dag\LL_j\right).
\]
\end{example}

To know what the actual state $\rho(t)$ of a QCTMC $\QC$ is,
one has to use the \emph{quantum (projective) measurement},
that is a collection of projectors $\PP_i$ with index $i$ taken from a finite set $I$,
satisfying $\sum_{i \in I} \PP_i = \id$.
After measurement,
the resulting state will be \emph{collapsed} to $\rho_i=\PP_i \rho\PP_i/p_i$
where $p_i=\tr(\PP_i \rho)$.
That is, quantum measurement destroys the states $\rho(t)$,
which is an important postulate of quantum mechanics~\cite[Section~2.2]{NiC00}.
To avoid such an unwanted change in $\rho(t)$,
we will adopt the manner of static analysis,
which uses some temporal logic to specify the properties of $\QC$
at a starting state $\rho(t_0)$ as in the following.

\subsection{Signal Temporal Logic}
Now we recall the syntax and the semantics of signal temporal logic (STL)
that is used to formally describe the repeated reachability properties in this paper.

\begin{definition}[\cite{MaN04}]\label{def:syntax}
	The syntax of the STL formulas are defined as follows:
	\[
	\Psi := \Phi \mid \neg\Psi \mid \Psi_1 \wedge \Psi_2 \mid \Psi_1 \ntl^\I \Psi_2
	\]
	in which the atomic propositions $\Phi$,
	interpreted as \emph{signals},
	are of the form $p(\x) \in \II$
	where $p$ is a $\mathbb{Q}$-polynomial in $\x=(x_1,\ldots,x_n)$
	with $x_i=\tr(\PP_i (\,\cdot\,))$ for some projector $\PP_i$
	and $\II$ is a rational value interval,
	and $\I$ is a finite time interval.
\end{definition}
As a special case, when $p$ is a convex combination over $\x$,
$p(\x)$ results in an \emph{observable} in quantum computing.
Here, $\ntl$ is called the until operator
and $\Psi_1 \ntl^\I \Psi_2$ is the until formula.

\begin{definition}\label{def:semantics}
	The semantics of the STL formulas are interpreted on the QCTMC $\QC$ in Definition~\ref{def:QCTMC}
	by the satisfaction relation $\models_\QC$ (or $\models$ for short):
	\[
	\begin{aligned}
		\rho(t) & \models \Phi
		&& \textup{if } p(\x) \in \II \textup{ holds for } x_i=\tr(\PP_i \cdot \rho(t)), \\
		\rho(t) & \models \neg\Psi
		&& \textup{if } \rho(t) \not\models \Psi, \\
		\rho(t) & \models \Psi_1 \wedge \Psi_2
		&& \textup{if } \rho(t) \models \Psi_1 \wedge \rho(t) \models \Psi_2, \\
		\rho(t) & \models \Psi_1 \ntl^\I \Psi_2
		&& \textup{if there exists a real number }t^* \in \I, \\
		&&& \textup{such that }\forall\,t' \in [t,t+t^*) \,:\, \rho(t') \models \Psi_1
		\textup{ and }\rho(t+t^*) \models \Psi_2.
	\end{aligned}
	\]
\end{definition}

The STL formulas are used to express a rich class of real-time properties.
For example,
the time-bounded reachability $\Diamond^\I\,\Psi_2$ is
a special case of the until formula $\Psi_1 \ntl^\I \Psi_2$
by setting $\Psi_1=\true$;
the time-bounded safety $\Box^\I\,\Psi$ amounts to $\neg\Diamond^\I\,(\neg\Psi)$.
So the repeated reachability in finite horizon, $\Box^\I\,\Diamond^\J\,\Phi$,
can also be expressed by the STL formula
$\neg \Diamond^\I\,(\neg \Diamond^\J\,\Phi) \equiv \neg (\true \ntl^\I (\neg (\true \ntl^\J \Phi)))$.

\section{Main Result}\label{S4}
In this paper,
we will study the repeated reachability problem in finite horizon, $\Box^\I\,\Diamond^\J\,\Phi$,
over the model of QCTMC $\QC$.
The outline of our approach is described in this section,
and more details will be provided in the coming sections.

We first rewrite the repeated reachability as
\[
	\Box^\I\,\Diamond^\J\,\Phi \equiv \neg\Diamond^\I\,(\neg\Diamond^\J\,\Phi),
\]
which is an STL formula and thus can be solved by the recent work~\cite{XMGY21}.
We give some hints on the \emph{solvability},
together with a review of the known approach.
\begin{enumerate}
	\item The \emph{instantaneous description} (ID) of $\QC$ can be obtained in polynomial time
		as a density operator function $\rho(t)$ w.\,r.\,t.\@ absolute time variable $t$.
		(Here, we introduce the fresh notion of IDs for dynamical states of $\QC$,
		which can be distinguished from the static states of the Hilbert space $\h$.)
	\item The atomic propositions $\Phi$ that represent signals in the real-time system
		are polynomial constraints in the outcome probabilities by projecting $\rho(t)$.
	\item We extract the exponential-polynomial (named \emph{observing} expression) $\phi(t)$
		from the signal $\Phi$,
		so that $\rho(0)$ meets $\Box^\I\,\Diamond^\J\,\Phi$
		if and only if $\phi(t)$ meets some sign-conditions,
		e.\,g.\@, $\phi>0$ or $\phi \ge 0$ holds in some appropriate time intervals just mentioned below.
	\item The post-monitoring period is determined as
	\[
		\BB_0:=[\inf\I+\inf\J,\sup\I+\sup\J],
	\]
		during which the sign information of $\phi(t)$ suffices to
		decide $\rho(0)\models\Box^\I\,\Diamond^\J\,\Phi$.
    \item After finding out all real roots $\lambda$ of $\phi(t)$ in $\BB_0$,
        we can obtain all solution time intervals $\delta_i$ during which $\rho(t)\models\Phi$ holds,
        whose endpoints are taken from the real roots $\lambda$ of $\phi(t)$.
        For each solution interval $\delta_i$,
        we have $\rho(t)\models\Diamond^\J\,\Phi$ holds
		for $t\in \I_i=\{t_1-t_2: t_1\in \delta_i \wedge t_2\in\J\}$.
	\item Hence we have reduced deciding $\rho(0) \models \Box^\I\,\Diamond^\J\,\Phi$
		to the real root isolation of $\phi(t)$ in $\BB_0$, a finite time interval.
		Particularly, the repeated reachability problem in finite horizon amounts to determining
        whether the union of all afore-calculated $\I_i$ covers $\I$.
\end{enumerate}

Then, we present new results for seeking more \emph{efficiency}.
A necessary and sufficient condition to $\rho(0) \models \Box^\I\,\Diamond^\J\,\Phi$
can be derived as
the existence of a finite collection $\mathbb{T}$ of absolute times $t^* \in \BB_0$,
satisfying that $\Phi$ holds at each $t^*$ and
the associated switch times
w.\,r.\,t.\@ the outer temporal operator $\Box^\I\,(\,\cdot\,)$
of $\Box^\I\,\Diamond^\J\,\Phi$ covers $\I$,
i.\,e.\@,
\[
	\I \subseteq \bigcup_{t^*\in\mathbb{T}} [t^*-\sup\J, t^*-\inf\J].
\]
Here, the time interval $\J$ is assumed to be closed for convenience.
Otherwise, we need to amend the intervals $[t^*-\sup\J, t^*-\inf\J]$
appearing in the RHS of the above inclusion
with appropriate endpoint conditions.
We employ a sample-driven procedure
by validating $\Phi$ with a few numerical samples $t^* \in \BB_0$.
After each time of sampling,
two straightforward criteria could be applied:
\begin{itemize}
	\item a successful sample $\rho(t^*)\models\Phi$ produces
		the segment $[t^*-\sup\J, t^*-\inf\J]$ (partially) covering $\I$;
	\item a conflicting sample $\rho(t^*)\not\models\Phi$ entails that
		$\Phi$ holds nowhere of a truth-invariant neighborhood $\delta$ of $t^*$.
		Thus we can safely exclude $\delta$ from the sample space $\BB$,
		which is initialized as the post-monitoring period $\BB_0$.
\end{itemize}
Repeat the sampling process until $\I$ has been completely covered
or the resulting $\BB$ is empty.
The termination is guaranteed when Schanuel's conjecture~\cite{Ax71} holds.
Checking $\rho(t^*)\models\Phi$ for \emph{concrete} $t^*$ is much cheap than
solving $\rho(t)\models\Phi$ w.\,r.\,t.\@ \emph{variable} $t$.
It is likely to yield an efficient decision procedure.
The improvement in intuition would be validated by randomly generated instances.

\section{Solvability by Real Root Isolation}\label{S5}
In this section,
we utilize the known results~\cite[Algorithm~1 \& Theorem~19]{XMGY21}
for solving the repeated reachability problem.
We first define two useful functions:
\begin{itemize}
	\item $\lv(\gamma):=\sum_{i=1}^d \sum_{j=1}^d \bra{i}\gamma\ket{j} \ket{i,j}$
	that rearranges entries of the linear operator $\gamma$
	on the Hilbert space $\h$ with dimension $d$
	as a $d^2$-dimensional column vector;
	\item $\vl(\mathbf{v}):= \sum_{i=1}^d \sum_{j=1}^d \bra{i,j} \mathbf{v} \op{i}{j}$
	that rearranges entries of the $d^2$-dimensional column vector $\mathbf{v}$
	as a linear operator on $\h$.
\end{itemize}
Here, $\lv$ and $\vl$ are pronounced
``linear operator to vector'' and ``vector to linear operator'', respectively.
They are mutually inverse functions,
so that
if a linear operator (resp.~its vectorization) is determined,
its vectorization (resp.~the original linear operator) is determined.
Hence, we can freely choose one of the two representations for convenience.
It is not hard to validate that for any linear operators $\mathbf{A}$, $\mathbf{B}$ and $\mathbf{C}$,
the matrix product $\mathbf{D}=\mathbf{A}\mathbf{B}\mathbf{C}$ has the transformation
\[
	\lv(\mathbf{D})=(\mathbf{A} \otimes \mathbf{C}^\T)\lv(\mathbf{B}),
\]
where $\T$ denotes transpose.

Based on the above notations and transformation,
the ID $\rho(t)$ characterized by the Lindblad's master equation~\eqref{eq:Lindblad}
can be rearranged as the ordinary differential equation
\begin{equation}\label{eq:ODE}
	\lv(\rho') = \M \cdot\lv(\rho),
\end{equation}
where
\[
	\M= -\imath\HH\otimes\id+\imath\id\otimes\HH^\T
	+\sum_{j=1}^m \left(\LL_j\otimes\LL_j^{*}  -\tfrac{1}{2}\LL_j^\dag\LL_j\otimes\id
	-\tfrac{1}{2}\id\otimes \LL_j^\T\LL_j^*\right)
\]
is called the \emph{governing} matrix for the Lindblad operator $\LLL$.
Its closed-form solution is given by
\begin{equation}
	\rho(t)=\exp(\LLL,t)(\rho(0))=\vl(\exp(\M\cdot t)\cdot\lv(\rho(0)))
\end{equation}
in standard literature, e.\,g.\@~\cite[Subsection~2.5.1]{Kai80},
and can be computed in polynomial time to get the explicit value.

To get information from a quantum system,
we would like to use a collection of projectors $\PP_1,\PP_2,\ldots,\PP_n$
to define the real-valued functions $x_i(t)=\tr(\PP_i \cdot \rho(t))$.
Namely, we have:
\begin{lemma}\label{lem:EP}
	Let $\rho(t)$ be the solution of the Lindblad's master equation~\eqref{eq:Lindblad},
	and $\PP$ a projector.
	Then $x(t)=\tr(\PP \cdot \rho(t))$ is a real-valued exponential polynomial 
	with the form
	\begin{equation}\label{eq:EP}
		\beta_1(t) \exp(\alpha_1 t) + \beta_2(t) \exp(\alpha_2 t) + \cdots + \beta_m(t) \exp(\alpha_m t),
	\end{equation}
	where $\alpha_1,\ldots,\alpha_m$ are distinct $\mathbb{A}$-numbers
	and $\beta_1(t),\ldots,\beta_m(t)$ are $\mathbb{A}$-polynomials.
\end{lemma}

Recall~\cite[Corollary~4.1.5]{Coh96} that
roots of the polynomials with algebraic coefficients ($\mathbb{A}$-polynomials)
are still algebraic numbers ($\mathbb{A}$-numbers).
So the lemma follows from the facts:
\begin{enumerate}
    \item The entries of $\HH$ and $\LL_j$ are $\mathbb{A}$-numbers,
    so are the entries of the governing matrix $\M$,
	implying that the eigenvalues of $\M$ are $\mathbb{A}$-numbers.
	\item The entries of $\rho(t)$ in closed form are exponential polynomials with the form~\eqref{eq:EP},
	as well as the entries of $\PP \cdot \rho(t)$ and $x(t)=\tr(\PP \cdot \rho(t))$.
	\item The Hermitian structure of $\rho(t)$ ensures that $x(t)$ is real-valued.
\end{enumerate}

\begin{example}\label{ex2}
    Here we continue to consider the QCTMC $\QC_1$ described in Example~\ref{ex1}.
    Let the initial ID $\rho(0)$ be $\op{00}{00}$.
    After solving the ordinary differential equation $\lv(\rho')=\M \cdot \lv(\rho)$
    where the governing matrix $\M$ is given by
\[
    -\imath\HH\otimes\id+\imath\id\otimes\HH^\T
	+\sum_{j=1}^2 \left(\LL_j\otimes\LL_j^{*} -\tfrac{1}{2}\LL_j^\dag\LL_j\otimes\id
	-\tfrac{1}{2}\id\otimes \LL_j^\T\LL_j^*\right),
\]
    we obtain the closed-form solution 
\[
\begin{aligned}
    \rho(t)=\ & [\tfrac{3}{8} +\tfrac{1}{4}\exp(-(2+2\imath) t)
        +\tfrac{1}{4}\exp(-(2-2\imath) t) +\tfrac{1}{8}\exp(-4 t)]\op{00}{00} \ + \\
    & [\tfrac{1}{8} -\tfrac{1}{4} \exp(-(2+2\imath) t)
        +\tfrac{1}{4} \exp(-(2-2\imath) t) - \tfrac{1}{8}\exp(-4 t)]\op{00}{11} \ + \\
    & [\tfrac{1}{8} +\tfrac{1}{4} \exp(-(2+2\imath) t)
        -\tfrac{1}{4} \exp(-(2-2\imath) t) - \tfrac{1}{8}\exp(-4 t)]\op{11}{00} \ + \\
    & [\tfrac{3}{8} -\tfrac{1}{4}\exp(-(2+2\imath) t)
        -\tfrac{1}{4}\exp(-(2-2\imath) t)+\tfrac{1}{8}\exp(-4 t) ]\op{11}{11}\ +\\
    & [\tfrac{1}{8}-\tfrac{1}{8}\exp(-4 t)](\op{01}{01}+\op{01}{10}+\op{10}{01}+\op{10}{10}).\\
\end{aligned}
\]
To get the probabilities of the two qubits staying respectively
in the basis states $\ket{00},\ket{01},\ket{10},\ket{11}$,
we choose the projectors $\PP_{i,j}=\op{i,j}{i,j}$ with $i,j\in\{0,1\}$, trace out $\rho(t)$,
and get
\[
\begin{aligned}
    x_1 & =\mathrm{tr}(\PP_{0,0}\cdot\rho(t))
    =\tfrac{3}{8} +\tfrac{1}{4}\exp(-(2+2\imath) t) +\tfrac{1}{4}\exp(-(2-2\imath) t)+\tfrac{1}{8}\exp(-4 t), \\
    x_2 & =\mathrm{tr}(\PP_{0,1}\cdot\rho(t))=\tfrac{1}{8}-\tfrac{1}{8}\exp(-4 t), \\
    x_3 & =\mathrm{tr}(\PP_{1,0}\cdot\rho(t))=\tfrac{1}{8}-\tfrac{1}{8}\exp(-4 t), \\
    x_4 & =\mathrm{tr}(\PP_{1,1}\cdot\rho(t))
    =\tfrac{3}{8}
    -\tfrac{1}{4}\exp(-(2+2\imath) t) -\tfrac{1}{4}\exp(-(2-2\imath) t)+ \tfrac{1}{8}\exp(-4 t),
\end{aligned}
\]
which will be used to make up the signals to be checked.
It is not hard to see that
all entries in $\rho(t)$ are exponential-polynomials with the form~\eqref{eq:EP}.
The same holds for $x_1$ through $x_4$,
which are additionally real-valued as $x_i^*=x_i$. \qed

\end{example}

Those exponential polynomials $x_1(t),x_2(t),\ldots,x_n(t)$
are basic ingredients to make up the atomic propositions $\Phi$ in STL.
Then, given an atomic proposition $\Phi \equiv p(\x) \in \II$
(assuming $\II$ is bounded),
we need to know the algebraic structure of the \emph{observing} expression
\begin{equation}\label{eq:EP1}
	\phi(t)=(p(\x(t))-\inf\II)(p(\x(t))-\sup\II),
\end{equation}
with which we will design an algorithm for solving the constraint $\Phi \equiv p(\x) \in \II$.
The structure of $\phi(t)$ depends on those of $x_i(t)=\tr(\PP_i \cdot \rho(t))$.
The latter $x_i(t)$ are exponential polynomials with the form~\eqref{eq:EP},
as well as $p(\x)$ and $\phi(t)$,
since they are polynomials in variables $\x$.
If $\II$ is unbounded from below (resp.~above),
the left (resp.~right) factor could be removed from~\eqref{eq:EP1} for further consideration.

Now we have reduced the truth of $\rho(t)\models \Phi$ with $\Phi\equiv p(\x) \in \II$
to determining the real roots of $\phi(t)$,
which can be completed by the real root isolation algorithm~\cite[Algorithm~1]{XMGY21}
\[
	\{\BB_1,\ldots,\BB_m\} \Leftarrow \textsf{Isolate}(\phi,\BB),
\]
in which the input $\phi(t)$ is a real-valued exponential polynomial defined on a rational interval $\BB=[l,u]$,
and the output $\BB_1,\ldots,\BB_m$ are finitely many disjoint intervals
such that each contains one real root of $\phi$ in $\BB$,
together contain all.

\begin{example}\label{ex3}
	Continuing to consider Example~\ref{ex2},
    we study the decision problem ---
    whether the atomic proposition $\Phi \equiv x_2-x_1^2>0$ 
	with $x_1 =\mathrm{tr}(\PP_{0,0}\cdot\rho(t))$
	and $x_2 =\mathrm{tr}(\PP_{0,1}\cdot\rho(t))$ holds for some time in $\BB=[0,3]$.
    The observing expression is
\[
\begin{aligned}
    \phi(t) =\ & x_2(t) - x_1^2(t) \\
    =\ & -\tfrac{1}{64} -\tfrac{3}{16}\exp(-(2+2\imath)t) -\tfrac{3}{16}\exp(-(2-2\imath)t)
    -\tfrac{1}{16}\exp(-(4+4\imath)t) \\
    & -\tfrac{1}{16}\exp(-(4-4\imath)t) -\tfrac{1}{16}\exp(-(6+2\imath)t)-\tfrac{1}{16}\exp(-(6-2\imath)t) \\
    & -\tfrac{11}{32}\exp(-4t) -\tfrac{1}{64}\exp(-8t).
\end{aligned}
\]
The polynomial representation of $\phi(t)$ is bivariated in $\exp(-t)$ and $\exp(\imath t)$,
as the numbers $1$ and $\imath$ in exponents are $\mathbb{Q}$-linearly independent.
Since $\phi(t)$ is irreducible,
it has neither rational root nor repeated root. 
After invoking~\cite[Algorithm~1]{XMGY21} on $\phi(t)$ with $\BB$,
we obtain two isolation intervals $[\tfrac{789}{800},\tfrac{1581}{1600}]$
(containing real root $\lambda_1 \approx 0.987368$, also see Fig.~\ref{fig1})
and $[\tfrac{39}{25},\tfrac{2499}{1600}]$ (containing $\lambda_2 \approx 1.56093$),
which could be easily refined up to any precision.
We can see that $(\lambda_1,\lambda_2) \cap \BB$ is a nonempty interval,
on which $\phi(t)$ is positive and $\Phi$ holds.
Hence the aforementioned decision problem is decided to be true. \qed

\begin{figure}[htb]
	\centering
\scalebox{0.9}{
\begin{tikzpicture}[spy using outlines=
{rectangle, magnification=3,rounded corners=3pt, connect spies}]
\tikzset{
every pin/.style={fill=yellow!50!white,rectangle,rounded corners=3pt,font=\tiny},
small dot/.style={fill=blue,circle,scale=0.1},
}
\begin{axis}[no markers,
every axis plot post/.append style={very thin},
ymin=-1, ymax=0.11,
axis y line=left,
clip=false,
axis x line=bottom,]
\addplot[blue, thin]table[x=x,y=y]{./data/data_ex1.txt};
\coordinate (spypoint) at (0.987368,0);
\coordinate (magnifyglass) at (0.987368,-0.4);
\coordinate (spypoint1) at (1.56093,0);
\coordinate (magnifyglass1) at (2,-0.6);
\draw[red,very  thin](0,0)--(3,0);
\node [small dot,pin=-60:{$0.987368$}] at (0.987368,0) {};
\node [small dot,pin=-60:{$1.56093$}] at (1.56093,0) {};
\end{axis}
\spy [blue,height = 0.5cm, width = 2cm] on (spypoint)
in node[fill=white, thin] at (magnifyglass);
\spy [blue,height = 0.5cm, width = 4.5cm] on (spypoint1)
in node[fill=white,thin] at (magnifyglass1);
\end{tikzpicture}
}
	\caption{Real roots of the observing expression $\phi(t)$}\label{fig1}
\end{figure}
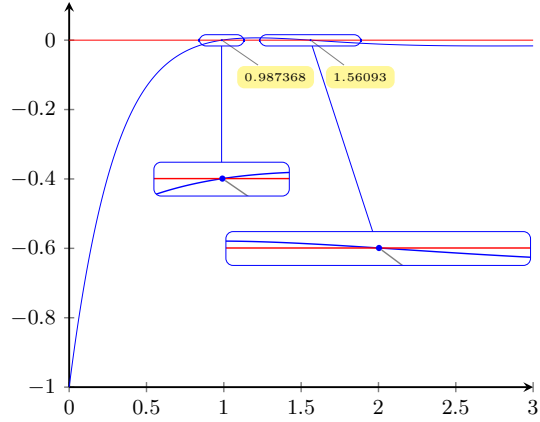
\end{example}

\begin{remark}
As a key step in the above isolation algorithm,
we need to compute $\mathbb{Q}$-linearly independent basis $\mu_1,\ldots,\mu_k$ of
the $\mathbb{A}$-exponents $\alpha_1,\ldots,\alpha_m$,
so that the original exponential polynomial $\phi(t)$ could be converted to a multivariate polynomial
that shares the same zeros.
The efficient Lenstra--Lenstra--Lovasz (LLL) algorithm~\cite{LLL82} could be applied here
to get the basis $\mu_1,\ldots,\mu_k$ of the Abelian group generated by $\alpha_1,\ldots,\alpha_m$.
\end{remark}

\begin{remark}\label{rmk2}
In fact, Schanuel's conjecture~\cite{Ax71} is a powerful tool to treat roots of
the general exponential polynomial.
For some special subclasses of exponential polynomials,
there are solid theorems to treat them:
one is Lindenman's theorem that has been employed in~\cite{AMW08}
for the exponential polynomials in the ring $\mathbb{Q}[t,\exp(t)]$,
the other is the Gelfond--Schneider theorem employed in~\cite[Subsection 4.1]{HLX+18}
for the exponential polynomials in $\mathbb{Q}[\exp(\mu_1t),\exp(\mu_2t)]$
where $\mu_1$ and $\mu_2$ are two $\mathbb{Q}$-linear independent real algebraic numbers.
However, they fail to isolate roots of elements in $\mathbb{Q}[t,\exp(\mu_1t),\ldots,\exp(\mu_kt)]$
for arbitrarily many $\mathbb{Q}$-linear independent complex algebraic numbers $\mu_1,\ldots,\mu_k$,
as considered in~\eqref{eq:EP}.
\end{remark}

For an STL formula $\Psi$,
independent from the starting time $t_0$,
the truth of $\rho(t_0)\models\Psi$ is affected
by the IDs $\rho(t)$ during a time period $t \in [t_0+0,t_0+\mnt(\Psi)]$.
We call it the \emph{post-monitoring} period where $\mnt(\Psi)$ is calculated as
\begin{equation}\label{eq:EP6}
	\begin{cases}
		0 & \textup{if }\Psi=\Phi, \\
		\mnt(\Psi_1) & \textup{if }\Psi=\neg\Psi_1, \\
		\max(\mnt(\Psi_1),\mnt(\Psi_2))	& \textup{if }\Psi=\Psi_1 \wedge \Psi_2, \\
		\sup\I+\max(\mnt(\Psi_1),\mnt(\Psi_2)) & \textup{if }\Psi=\Psi_1 \ntl^\I \Psi_2.
	\end{cases}
\end{equation}
For the repeated reachability $\Box^\I\,\Diamond^\J\,\Phi$,
the post-monitoring period can be simplified to $[\inf\I+\inf\J,\sup\I+\sup\J]$,
which will be the input time interval $\BB$ of the above isolation algorithm $\textsf{Isolate}(\phi,\BB)$.

Once all real roots $\lambda$ of $\phi(t)$ in $\BB$ are obtained,
we can determine all solutions during which $\rho(t)\models\Phi$ holds,
which are delivered as finitely many intervals $\delta_i$
with endpoints taken from those real roots $\lambda$ of $\phi(t)$.
For each solution interval $\delta_i$,
we have $\rho(t)\models\Diamond^\J\,\Phi$ holds for $t\in \I_i=\{t_1-t_2: t_1\in \delta_i \wedge t_2\in\J\}$,
a coverage.
Further, if the coverage union $\bigcup_i \I_i$ over all solution intervals $\delta_i$ completely covers $\I$,
the repeated reachability problem in finite horizon,
i.\,e.\@ $\rho(0) \models \Box^\I\,\Diamond^\J\,\Phi$,
can be decided to be true;
otherwise be false.

\begin{example}\label{ex4}
    Consider the repeated reachability property $\Psi\equiv\Box^\I\,\Diamond^\J\,\Phi$
	on the QCTMC $\QC_1$ shown in Example~\ref{ex1},
	where $\I=[0,\tfrac{3}{2}]$ and $\J=[0,1]$ are time intervals,
	and $\Phi \equiv x_2-x_1^2>0$ is an atomic proposition.
    The post-monitoring period $\mnt(\Psi)$ is $\sup\I+\sup\J=\tfrac{5}{2}$ obtained by Eq.~\eqref{eq:EP6}.
	Thereby it suffices to study the behavior of $\QC_1$ during $\BB = [0,\tfrac{5}{2}]$.
    From Example~\ref{ex3} we have known that
    $\rho(t) \models\Phi$ holds for $t\in(\lambda_1,\lambda_2)$,
    where $\lambda_1 \approx 0.987368$ and $\lambda_2 \approx 1.56093$ are two real roots of $x_2-x_1^2$.
    It implies that the ID $\rho(t)$ meets $\Diamond^\J\,\Phi$ when $t\in(\lambda_1-\sup\J,\lambda_2-\inf\J)$.
    Hence the repeated reachability property $\Psi$ is decided to be true at the initial ID,
    i.\,e.\@ $\rho(0) \models \Psi$,
    as $(\lambda_1-\sup\J,\lambda_2-\inf\J)$ completely covers $\I$. \qed
\end{example}

Finally we summarize the above as:
\begin{theorem}[Decidability]\label{thm:decide}
	The repeated reachability problem in finite horizon is decidable on quantum continuous-time Markov chains,
    when Schanuel's conjecture holds.
\end{theorem}
Although this result is conditional,
some unconditional results can be obtained for the subclasses mentioned in Remark~\ref{rmk2}.
We recall and refine the essentials of the construction
particular for the repeated reachability in this section,
which will be used to make up a more efficient solving procedure soon.

\begin{remark}
	We have to point out that the boundedness of the time intervals $\I$ and $\J$ is a real restriction,
	without which we need to develop the real root isolation of $\phi(t)$ during an unbounded time interval.
	That goes beyond the rich scope of order-minimal theory~\cite{Wil96}
	that admits the solvability for any real analytic function restricted in a bounded region,
	and thus bring technical hardness.
	However, it is not in the case when we deal with
	the repeated reachability in finite horizon $\Box^\I\,\Diamond^\J\,\Phi$
	where the time intervals $\I$ and $\J$ are bounded.	
\end{remark}

\section{Efficiency Driven by Samples}\label{S6}
In the previous section we have shown the decidability of the repeated reachability in finite horizon.
Now we are to design a more efficient sample-driven solving procedure.
Before describing its rationale, we need a technical tool.

\begin{lemma}\label{lem:sign}
	Let $\rho(t)$ be the ID function of a QCTMC $\QC$ w.\,r.\,t.\@ time variable $t$,
	and $\Phi$ an STL formula.
	Then $\rho(t^*)\models\Phi$ is decidable at any rational sample $t^*$.
\end{lemma}
\begin{proof}
	We notice that
	$\rho(t)\models\Phi$ holds if and only if
	the observing expression $\phi(t)$ meets some sign-condition.
	Putting the rational number $t=t^*$ into Eq.~\eqref{eq:EP},
	we can see that $\phi(t^*)$ is a real number with the form
\[
	\beta_0 + \beta_1 \e^{\alpha_1} + \cdots + \beta_m \e^{\alpha_m},
\]
	where $\alpha_1,\ldots,\alpha_m$ are distinct $\mathbb{A}$-numbers
	and $\beta_0,\ldots,\beta_m$ are $\mathbb{A}$-numbers.
	By the fact \cite[Theorem~1.2]{Bak75} that $\e$ is \emph{transcendental} (i.\,e.\@ not algebraic),
	we can infer that $\phi(t^*)=0$ if and only if all $\beta_0,\ldots,\beta_m$ are zero.
	If $\phi(t^*)\ne 0$, the sign of $\phi(t^*)$ can be determined by sufficiently approaching $\e$,
	thus $\rho(t^*)\models\Phi$ is decided. \qed
\end{proof}

We turn to describe the rationale of the sample-driven solving procedure.
If an ID $\rho(t^*)$ of $\QC$ meets the STL formula $\Phi$,
whose decidability has just been proved in Lemma~\ref{lem:sign},
the initial ID $\rho(0)$ meets the STL formula $\Box^{\I'}\,\Diamond^\J\,\Phi$
for the interval $\I'=[t^*-\sup\J,t^*-\inf\J]$.
So, the repeated reachability $\rho(0) \models \Box^\I\,\Diamond^\J\,\Phi$ can be inferred from
the existence of a finite collection $\mathbb{T}$ of absolute times $t^*$,
satisfying $\rho(t^*) \models \Phi$,
over which the union $\bigcup_{t^*} [t^*-\sup\J,t^*-\inf\J]$ covers $\I$.
Moreover, the distance between two successive samples $t_i^*$ and $t_{i+1}^*$ in $\mathbb{T}$
should be not greater than the length $|\J|$ of $\J$,
as otherwise $[t_i^*-\sup\J,t_i^*-\inf\J]$ and $[t_{i+1}^*-\sup\J,t_{i+1}^*-\inf\J]$ are disjoint
and cannot cover the connected interval $\I$.
The existence of such a collection $\mathbb{T}$ is not only a sufficient condition but also a necessary one,
which is revealed by:
\begin{lemma}\label{lem:finite}
	There is a finite collection $\mathbb{T}$ of absolute times $t^*$,
	satisfying $\rho(t^*) \models \Phi$,
	over which the union
	\[
		\bigcup_{t^*\in\mathbb{T}} [t^*-\sup\J,t^*-\inf\J]
	\]
	covers $\I$,
	provided that $\rho(0) \models \Box^\I\,\Diamond^\J\,\Phi$ holds.
\end{lemma}
\begin{proof}
Let $\mathbb{S}$ be the (possibly infinite) collection of those absolute times $t^* \in \BB_0$
with $\rho(t^*) \models \Phi$.
Under the assumption $\rho(0) \models \Box^\I\,\Diamond^\J\,\Phi$,
we know
\[
    \I \subseteq \bigcup_{t^*\in\mathbb{S}} [t^*-\sup\J,t^*-\inf\J].
\]
We proceed to refine $\mathbb{S}$ to a finite collection $\mathbb{T}$, satisfying
\[
    \I \subseteq \bigcup_{t^*\in\mathbb{T}} [t^*-\sup\J,t^*-\inf\J].
\]
The observing expression $\phi(t)$ extracted from the signal $\Phi$ is a real analytic function
that has finitely many real roots during any compact interval $\BB$, saying $\BB_0$.
So $\mathbb{S}$ has the union structure of
finitely many open intervals plus finitely many singleton sets,
on each of which $\phi(t)$ is sign-invariant.
It entails that $\mathbb{S}$ consists of finitely many maximal connected intervals $\mathbb{S}_j$,
no matter whether they are closed or open.
Correspondingly, there are finitely many closed intervals $\I_j$,
not necessarily disjoint pairwise,
such that
\[
	\I_j \subseteq \I \cap \left(\bigcup_{t^*\in\mathbb{S}_j} [t^*-\sup\J,t^*-\inf\J]\right)
\]
and $\I=\bigcup_j \I_j$.
Since the closed interval $\I_j$ has length $|\I_j|$
and every $\rho(t^*) \models \Phi$ gives a coverage with length $|\J|$ to $\I$,
we can refine $\mathbb{S}_j$ to at most $1+\lfloor |\I_j|/|\J| \rfloor$ elements $\mathbb{T}_j$ of $\mathbb{S}_j$
to achieve the same coverage.
Thereby, we get the desired finite collection $\mathbb{T}=\bigcup_j \mathbb{T}_j$
to achieve the same coverage as $\mathbb{S}$ does. \qed
\end{proof}

As analyzed above, we have known that
a successful sample $\rho(t^*)\models\Phi$ gives rise to
the segment $[t^*-\sup\J, t^*-\inf\J]$ partially covering $\I$.
However, what can we learn from a conflicting sample $\rho(t^*)\not\models\Phi$?
It is a neighborhood $\delta$ of $t^*$,
in which the observing expression $\phi(t)$ is sign-invariant.
It entails that if $\rho(t^*)\models\Phi$ does not hold,
$\rho(t)$ does not meet $\Phi$ anywhere of $\delta$,
and thus we can safely exclude this $\delta$ from the sample space $\BB$.

In detail, after a trial sample $t^*$,
no matter whether it is successful or conflicting,
we can calculate a sign-invariant neighborhood $\delta$ with the following manner.
If $\phi(t^*)=0$, it is an \emph{equation-type} constraint.
So we set $\delta$ to be the singleton set $\{t^*\}$.
Otherwise, it is an \emph{inequality-type} constraint.
Let $\psi_1,\linebreak[0]\ldots,\psi_k$ be all distinct irreducible factors of $\phi$,
and $\psi_1',\linebreak[0]\ldots,\psi_k'$ their respective derivatives.
\begin{itemize}
\item Firstly, we compute the radius $\epsilon_j = |\psi_j(t^*)|/\sup_{t\in\BB} |\psi_j'(t)|$.
Here, the numerator $|\psi_j(t^*)|$ is the height of $\psi_j(t^*)$ from zero
while the denominator $\sup_{t\in\BB} |\psi_j'(t)|$ is the maximal change rate of $\psi_j(t)$
during the sample space $\BB$,
so that $\psi_j(t)$ is sign-invariant in $(t^*-\epsilon_j,t^*+\epsilon_j)$.
\item Secondly, we compute the radius $\theta_j = |\psi_j'(t^*)|/\sup_{t\in\BB} |\psi_j''(t)|$,
so that $\psi_j(t)$ is monotonous in $(t^*-\theta_j,t^*+\theta_j)$.
Whenever $\theta_j>\epsilon_j$,
we could get a sign-invariant open interval $\delta_j$ extending $(t^*-\epsilon_j,t^*+\epsilon_j)$.
It is achieved by setting the left endpoint
\begin{equation}\label{eq:neighbor}
	\inf\delta_j=\begin{cases}
		t^*-\epsilon_j & \textup{if } \theta_j \le \epsilon_j, \\
		t^*-\theta_j 
		& \textup{if } \theta_j > \epsilon_j \wedge \psi_j(t^*) \psi_j(t^*-\theta_j) > 0, \\
		s^* & \textup{if } \theta_j > \epsilon_j \wedge \psi_j(t^*) \psi_j(t^*-\theta_j) \le 0,
	\end{cases}
\end{equation}
where $s^*$ is the unique zero of $\phi(t)$ in $(t^*-\theta_j,t^*)$,
which can be efficiently approached by monotonicity.
The right endpoint $\sup\delta_j$ is set symmetrically.
\item Finally, we obtain the neighborhood $\delta=\bigcap_{j=1}^k \delta_j$ to be excluded.
All factors $\psi_j(t)$ are sign-invariant in $\delta_j$,
so is their product $\phi(t)$ in $\delta$.
It also implies that $\rho(t)$ is truth-invariant in $\delta$.
\end{itemize}

\begin{example}\label{ex5}
    Again, we consider the repeated reachability property $\Psi\equiv\Box^\I\,\Diamond^\J\,\Phi$
	as in Example~\ref{ex4}.
    Here we will show how to compute the sign-invariant neighborhood $\delta$ of some concrete samples $t^*$.
    Suppose that we are given the sample $t_1^*=\tfrac{6}{5}$,
	at which the observing expression $\phi(\tfrac{6}{5})\approx 0.0066092$ is positive
	(see Fig.~\ref{fig1})
	and thus the ID $\rho(\tfrac{6}{5})$ meets $\Phi$.
	Since $\phi(t)$ itself is irreducible,
	we use the single exponential polynomial to compute the radius of the sign-invariant neighborhood.
	First, it is not hard to compute the two bounds
	$|\phi'(t)|<\tfrac{7}{2}$ and $|\phi''(t)|<\tfrac{21}{2}$ whenever $t\in\BB=[0,\tfrac{5}{2}]$.
	We then get
	\[
		\epsilon_1 =\left|\phi\left(\tfrac{6}{5}\right)\right|/\tfrac{7}{2} \gtrsim \tfrac{9441}{5000000}
		\quad \textup{and} \quad
		\theta_1 =\left|\phi'\left(\tfrac{6}{5}\right)\right|/\tfrac{21}{2} \gtrsim \tfrac{14689}{50000000}.
	\]
	As $\theta_1<\epsilon_1$, 
	we calculate the sign-invariant neighborhood $\delta_1$ as
	$(t_1^*-\epsilon_1, t_1^*+\epsilon_1) = (\tfrac{5990559}{5000000}, \tfrac{6009441}{5000000}) \linebreak[0]
	\approx (1.198112, 1.20189)$,
	which is excluded for consideration.

	We consider another sample $t_2^*=\tfrac{99}{100}$,
	at which the ID $\rho(\tfrac{99}{100})$ also satisfies the formula $\Phi$
	as $\phi(\tfrac{99}{100}) \approx 0.000017626>0$.
	Similarly, we get
	\[
		\epsilon_2 =\left|\phi\left(\tfrac{99}{100}\right)\right|/\tfrac{7}{2} \gtrsim \tfrac{1259}{25000000}
		\quad \textup{and} \quad
		\theta_2 =\left|\phi'\left(\tfrac{99}{100}\right)\right|/\tfrac{21}{2} \gtrsim \tfrac{1581}{250000}.
	\]
	As $\theta_2>\epsilon_2$,
	the sign-invariant neighborhood is $\delta_2 =(s_1^*,t_2^*+\theta_2)$,
	since $\phi(t_2^*)\phi(t_2^*-\theta_2)$ is negative while $\phi(t_2^*)\psi(t_2^*+\theta_2)$ is positive.
	Here $s_1^*=\lambda_1 \approx 0.987368$ is the unique real root of $\phi(t)$ during $(t_2^*-\theta_2,t_2^*)$,
	which could be approached up to any precision. \qed
\end{example}

For a successful sample $\rho(t^*) \models \Phi$,
we could further speed up the solving procedure by
expanding the coverage $[t^*-\sup\J,t^*-\inf\J]$
to the theoretically perfect $(\inf\delta-\sup\J,\sup\delta-\inf\J)$,
since $\phi(t)$ is sign-invariant in the neighborhood $\delta$.
But we fall to sample at the endpoints of $\delta$,
as the interval $\delta$ is open.
For the sake of effectiveness,
we need to provide only finitely many samples $\mathbb{T}$ from $\delta$
to cover $(\inf\delta-\sup\J,\sup\delta-\inf\J)$ as \emph{essentially} as possibly.
Here the term `essentially' means that
the missed coverage should not be too much,
saying Lebesgue measure not greater than $|\J|:=\sup\J-\inf\J$.
Assuming $|\delta|>|\J|$ (and otherwise trivially),
it can be achieved by three steps.
\begin{enumerate}
	\item The leftmost sample $l$ is chosen to
	be any element in $(\inf\delta,\inf\delta+|\J|/2]$.
	\item The rightmost sample $u$ is chosen to
	be any element in $[\sup\delta-|\J|/2,\sup\delta)$.
	\item The intermediate samples between $l$ and $u$ are
	any arithmetic progression with common difference not greater than $|\J|$.
\end{enumerate}
In the above treatment, for two successive samples $t_i^*$ and $t_{i+1}^*$,
we omit the intermediate samples $t\in(t_i^*,t_{i+1}^*)$ to produce a coverage,
which has already been produced by $t_i^*$ and $t_{i+1}^*$,
i.\,e.\@,
\begin{equation}\label{eq:essential1}
\begin{aligned}
	& \left([t_i^*-\sup\J,t_i^*-\inf\J] \cup [t_{i+1}^*-\sup\J,t_{i+1}^*-\inf\J]\right) \\
	= \ & \bigcup_{t\in[t_i^*,t_{i+1}^*]} [t-\sup\J,t-\inf\J],
\end{aligned}
\end{equation}
since the distance between $t_i$ and $t_{i+1}$ is bounded by $|\J|$.
It is also illustrated by Fig.~\ref{fig:cover1}.
There, in the first line, we can see that
the distance between two successive samples $t_i^*$ and $t_{i+1}^*$ is not greater than $|\J|$.
The two samples $t_i^*$ and $t_{i+1}^*$ respectively produce a coverage with length $|\J|$ in the second line.
The second line also shows that their coverages must be connected.

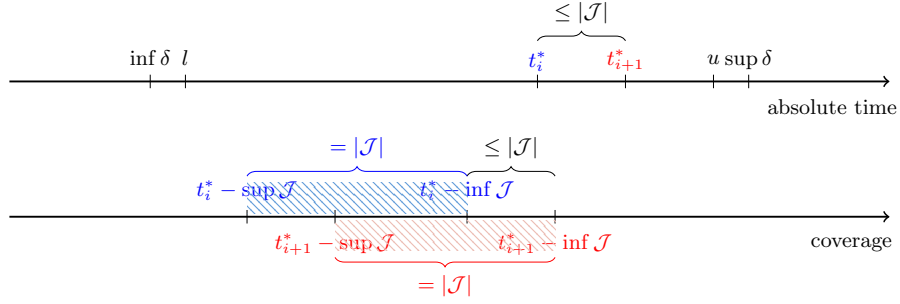
\begin{figure}[htb]
    \centering
    \scalebox{0.9}{
\begin{tikzpicture}
    
    \draw[thick,->] (0,0) coordinate (O2) -- (13,0) coordinate (E2);

    \coordinate (atime) at ($(O2)!0.935!(E2)$);
    \coordinate (infdelta2) at ($(O2)!0.16!(E2)$);
    \coordinate (l2) at ($(O2)!0.2!(E2)$);
    \coordinate (tistar) at ($(O2)!0.6!(E2)$);
    \coordinate (ti1star) at ($(O2)!0.7!(E2)$);
    \coordinate (u2) at ($(O2)!0.8!(E2)$);
    \coordinate (supdelta2) at ($(O2)!0.84!(E2)$);
    \foreach \x/\y in {infdelta2/$\inf\delta$, l2/$l$, tistar/$ $, ti1star/$ $, u2/$ $,supdelta2/$ $}
	{\draw (\x)--+(0, 1mm); \draw (\x)--+(0, -1mm);\node[above] at ([yshift=1.3mm]\x) {\y};}
     \node[above] at ([yshift=0.4mm]tistar) {\textcolor{blue}{$t_i^*$}};
     \node[above] at ([yshift=0.4mm] ti1star) {\textcolor{red}{$t_{i+1}^*$}};
     \node[above] at ([yshift=0.7mm] supdelta2) {\textcolor{black}{$\sup\delta$}};
     \node[above] at ([yshift=1.4mm] u2) {\textcolor{black}{$u$}};
    \node[below] at ([yshift=-1.5mm] atime) {absolute time};
    \mybracesmall{tistar}{ti1star}{$\leq|\J|$}{black}
    
    \draw[thick,->] (0,-2) coordinate (O1) -- (13,-2) coordinate (E1);
    \coordinate (coverage) at ($(O1)!0.957!(E1)$);
    \coordinate (tisupj) at ($(O1)!0.27!(E1)$);
    \coordinate (ti1supj) at ($(O1)!0.37!(E1)$);
    \coordinate (tiinfj) at ($(O1)!0.52!(E1)$);
    \coordinate (ti1infj) at ($(O1)!0.62!(E1)$);
    
    \mybrace{tiinfj}{ti1infj}{$\leq|\J|$}{black}
    \foreach \x/\y in {tisupj/$t_i^*-\sup\J$, tiinfj/$t_i^*-\inf\J$}
	{\draw (\x)--+(0, -1mm);\draw (\x)--+(0, +1mm);\node[above] at ([yshift=1.3mm] \x){\textcolor{blue}{\y}};}
    
    \foreach \x/\y in {ti1supj/$t_{i+1}^*-\sup\J$, ti1infj/$t_{i+1}^*-\inf\J$}
	{\draw (\x)--+(0, -1mm);\draw (\x)--+(0, +1mm);\node[below] at ([yshift=-1.3mm]\x){\textcolor{red}{\y}};}

    \node[below] at ([yshift=-2mm]coverage) {coverage};
    \mybracesmall{tisupj}{tiinfj}{$=|\J|$}{blue}
    \mybracesmallbelow{ti1supj}{ti1infj}{$=|\J|$}{red}
    \path[pattern=north west lines,pattern color=pink1] ([yshift=-0.5mm] ti1supj) rectangle ([yshift=-5mm] ti1infj)--cycle;
     \path[pattern=north west lines, pattern color=blue1] ([yshift=0.5mm] tisupj) rectangle ([yshift=5mm] tiinfj)--cycle;

\end{tikzpicture}
}
    \caption{Coverage produced by two successive samples $t_i^*$ and $t_{i+1}^*$}\label{fig:cover1}
\end{figure}

Moreover, we omit the leftmost segment $(\inf\delta,l)$
(and the rightmost segment $(u,\sup\delta)$) to produce a coverage.
It will not cause any trouble for the following reason.
The trouble occurs only when
the current neighborhood $\delta$ cannot produce a coverage containing $(\inf\delta-\sup\J,l-\sup\J)$,
implying that $\inf\delta-\sup\J$ should be covered
since the target $\I$ is compact (bounded and closed).
So there must be a neighborhood $\delta_\textup{left}$ left to $\delta$
that produces the coverage containing $\inf\delta-\sup\J$.
It entails that $\delta_\textup{left}$ contains $\inf\delta$.
Using the same treatment, we can see that
the distance between the rightmost sample $u_\textup{left}$ of $\delta_\textup{left}$
and the leftmost sample $l$ of $\delta$ is not greater than
$(\sup\delta_\textup{left}-u_\textup{left})+(l-\inf\delta) \le |\J|/2+|\J|/2=|\J|$,
and thus produces the coverage as $(\inf\delta,l)$,
i.\,e.\@,
\begin{equation}\label{eq:essential2}
\begin{aligned}
	& \left([u_\textup{left}-\sup\J,u_\textup{left}-\inf\J] \cup [l-\sup\J,l-\inf\J]\right) \\
	\supset \ & \bigcup_{t\in(\inf\delta,l)} [t-\sup\J,t-\inf\J].
\end{aligned}
\end{equation}
It is illustrated by Fig.~\ref{fig:cover2}.
There, the leftmost segment $(\inf\delta,l)$ of $\delta$ has length not greater than $|\J|/2$;
the rightmost segment $(u_\textup{left},\sup\delta_\textup{left})$ of $\delta_\textup{left}$
also has length not greater than $|\J|/2$.
The two segments are overlapping,
so the distance between $u_\textup{left}$ and $l$ is not greater than $|\J|$.

\begin{figure}[htb]
    \centering
    \scalebox{0.9}{
\begin{tikzpicture}
    \draw[thick,->] (0,0) coordinate (O2) -- (13,0) coordinate (E2);
    \coordinate (coverage) at ($(O2)!0.935!(E2)$);
    \coordinate (infdeltaleft) at ($(O2)!0.04!(E2)$);
    \coordinate (lleft) at ($(O2)!0.13!(E2)$);
    \coordinate (uleft) at ($(O2)!0.43!(E2)$);
    \coordinate (supdeltaleft) at ($(O2)!0.53!(E2)$);
     \coordinate (infdelta) at ($(O2)!0.48!(E2)$);
    \coordinate (l) at ($(O2)!0.58!(E2)$);
    \coordinate (supdelta) at ($(O2)!0.83!(E2)$);
    \coordinate (u) at ($(O2)!0.78!(E2)$);
    \foreach \x/\y in {infdeltaleft/$\inf\delta_{\textup{left}}$, lleft/$l_{\textup{left}}$,
	supdeltaleft/$ $, uleft/$ $}
	{\draw (\x)--+(0, 1mm);\draw (\x)--+(0, -1mm);\node[above] at ([yshift=1.3mm] \x) {\textcolor{blue}{\y}};}
    \node[above] at ([yshift=1.2mm] supdeltaleft) {\textcolor{blue}{$\sup\delta_{\textup{left}}$}};
    \node[above] at ([yshift=1.8mm] uleft) {\textcolor{blue}{$u_{\textup{left}}$}};
    \foreach \x/\y in {infdelta/$\inf\delta$, l/$l$, supdelta/$ $, u/$ $}
	{\draw (\x)--+(0, 1mm);\draw (\x)--+(0, -1mm);\node[below] at ([yshift=-1.3mm] \x) {\textcolor{red}{\y}};}
    \node[below] at ([yshift=-2mm] u) {\textcolor{red}{$u$}};
    \node[below] at ([yshift=-1.5mm] supdelta) {\textcolor{red}{$\sup\delta$}};
    \node[below] at ([yshift=-1.1mm] coverage) {absolute time};

    \mybracesmall{uleft}{supdeltaleft}{$\le|\J|/2$}{blue}
    \mybracesmallbelow{infdelta}{l}{$\le|\J|/2$}{red}
 \path[pattern=north west lines,pattern color=pink1] ([yshift=-0.5mm] infdelta) rectangle ([yshift=-5mm] supdelta)--cycle;
     \path[pattern=north west lines,pattern color=blue1] ([yshift=0.5mm] infdeltaleft) rectangle ([yshift=5mm] supdeltaleft)--cycle;
\end{tikzpicture}
}
    \caption{Positions of the current neighborhood $\delta$
	and the left one $\delta_\textup{left}$}\label{fig:cover2}
\end{figure}
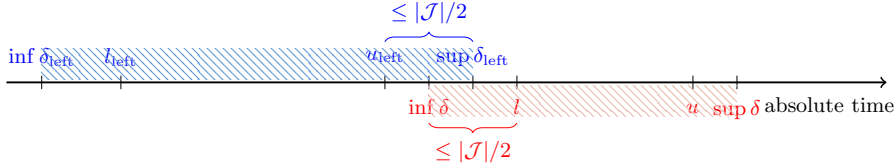

\begin{example}\label{ex6}
Continuing to consider Example~\ref{ex5},
we now are to decide
whether the initial ID $\rho(0)$ meets the repeated reachability property $\Box^\I\,\Diamond^\J\,\Phi$.
It is solved by the sample-driven solving procedure shown in Table~\ref{table:ex6}.
Let us explain it below.
For the sample $t_3^*=1$ at which the atomic proposition $\Phi$ holds,
i.\,e.\@, $t_3^*$ is a satisfying sample,
its sign-invariant neighborhood $\delta_3$ is $(\tfrac{4970403}{5000000}, \tfrac{5029597}{5000000})$,
the essential samples from $\delta_3$ is $\{1\}$,
which contributes the coverage $\I_3=[1-\sup\J,1-\inf\J]=[0,1]$ to the target $\I$.
For another sample $t_4^*=\tfrac{3}{2}$ at which $\Phi$ is met too,
its sign-invariant neighborhood $\delta_4$ is $(\frac{748689}{500000}, \tfrac{ 751311}{500000})$,
the essential samples from $\delta_4$ is $\{\tfrac{3}{2}\}$,
which contributes the coverage $\I_4=[\tfrac{1}{2},\tfrac{3}{2}]$ to $\I$.
Totally we get the finite collection of satisfying samples $\mathbb{T}=\{1,\tfrac{3}{2}\}$,
so that the time interval $\I=[0,\tfrac{3}{2}]$ is covered by $\I'=\I_3 \cup \I_4=[0,\tfrac{3}{2}]$.
Hence the repeated reachability $\Box^\I\,\Diamond^\J\,\Phi$ holds at the initial ID,
i.\,e.\@ $\rho(0) \models \Box^\I\,\Diamond^\J\,\Phi$. \qed

\begin{table}[htb]
\caption{Sample-driven procedure for deciding $\Box^\I\,\Diamond^\J\,\Phi$}\label{table:ex6}
\centering
{\small
\begin{tabular}{p{2cm}<{\centering}cccp{2cm}<{\centering}}
\toprule
samples $t^*$ & radius & neighborhood $\delta$ & ess. samples $\mathbb{T}$ & coverage $\I'$ \\
\midrule
$t_3^*=1$ & \multirow{2}{*}{$\theta_3=\tfrac{591949}{100000000}$}
& $\delta_3=(t_3^*-\theta_3, t_3^*+\theta_3)$ & \multirow{2}{*}{$\{1\}$}& \multirow{2}{*}{$[0,1]$} \\
satisfying && $\approx (0.994081,1.00591)$ && \\
\midrule
$t_4^*=\tfrac{3}{2}$ & \multirow{2}{*}{$\theta_4=\tfrac{1311031}{500000000}$}
& $\delta_4=(t_4^*-\theta_4, t_4^*+\theta_4)$ & \multirow{2}{*}{$\{\tfrac{3}{2}\}$}
& \multirow{2}{*}{$[\tfrac{1}{2},\tfrac{3}{2}]$} \\
satisfying && $\approx (1.497378, 1.50262)$ && \\
\bottomrule
\end{tabular}}
\end{table}
\end{example}

We summarize the solving procedure as Algorithm~\ref{solve}
with correctness analysis delivered here.
We first notice that
the sample space $\BB$ contains all satisfying samples
to the repeated reachability problem $\Box^\I\,\Diamond^\J\,\Phi$.
These satisfying samples can be split into sign-invariant neighborhoods $\delta$ of concrete samples $t^*$
by the fact the observing expression $\phi(t)$ is real analytic on the compact interval $\BB$.
For each neighborhood $\delta$,
there are at most $1+\lfloor|\delta|/|\J|\rfloor$ essential samples
that contribute the necessary and sufficient coverage,
as revealed by Eqs.~\eqref{eq:essential1} and~\eqref{eq:essential2}.
On the other hand,
there are finitely many neighborhoods $\delta$ computed by Eq.~\eqref{eq:neighbor} in the procedure,
since otherwise there is an irreducible factor $\psi_j(t)$ of $\phi(t)$,
such that for any positive constant $\epsilon$,
$|\psi_j(t)| \le \epsilon$ and $|\psi_j'(t)| \le \epsilon$
simultaneously hold at some sample $t^*$ in $\BB$.
The latter will lead to a contradiction as follows.
\begin{enumerate}
\item Let $\chi=\{t^*\in \BB: \psi_j(t^*)=0\}$ be the set of real roots of $\psi_j(t)$ in $\BB$.
It is a finite set, since $\psi_j(t)$ is analytic on $\BB$ and $\BB$ is compact.
\item Let $\epsilon_1=\min \{|\psi_j'(t^*)|:t^* \in \chi\}$.
It is a positive constant,
since otherwise the irreducible factor $\psi_j(t)$ has a repeated real root,
which contradicts Schanuel's conjecture~\cite{Ax71}
stating that an irreducible exponential polynomial has no repeated root with the only possible exception $0$.
(The conjecture is commonly believed
to be an unproved theorem by mathematical community,
since it has been raised in 1960's.
Once the conjecture fails at any instance to Algorithm~\ref{solve},
an important breakthrough in number theory will be achieved by that instance serving as a counterexample.
Up to the present, no such instance has been reported.)
\item Let $\delta(t,r)$ denote the neighborhood of $t$ with radius $r$.
We choose $r_0$ to be such a radius,
satisfying that $|\psi_j'(t)|$ is not less than $\tfrac{1}{2}\epsilon_1$
on all neighborhoods $\delta(t^*,r_0)$ with $t^*\in \chi$.
Let $\chi^o=\bigcup_{t^*\in \chi} \delta(t^*,r_0)$
be the union of finitely many neighborhoods.
So we have that $|\psi_j'(t)| \ge \tfrac{1}{2}\epsilon_1$ holds on $\chi^o$.
\item Let $\epsilon_2=\inf\{|\psi_j(t)|:t \in \BB \setminus \chi^o\}$.
It is a positive constant,
since all real roots $t^*$ of $\psi_j(t)$ in $\BB$ are excluded with their neighborhoods $\delta(t^*,r_0)$.
That means $|\psi_j(t)| \ge \epsilon_2$ holds on $\BB \setminus \chi^o$.
\end{enumerate}
Combining the last two issues,
we obtain that for each sample in $\BB$,
either $|\psi_j(t)| \ge \tfrac{1}{2}\epsilon_1$ holds or $|\psi_j'(t)| \ge \epsilon_2$ holds,
and that the procedure invokes Eq.~\eqref{eq:neighbor} to compute the neighborhoods $\delta$
only finitely many times.
Hence the output $\mathbb{T}$ is a finite collection of samples that produce the same coverage
as the whole sample space $\BB$ does,
entailing Algorithm~\ref{solve} is correct.

\begin{algorithm}[htb]
	\caption{\textsf{A Sample-Driven Solving Procedure}}\label{solve}
	\begin{algorithmic}[1]
		\item[] $$\{t_1^*,\ldots,t_m^*\} \Leftarrow \textsf{Solve}(\rho,\Box^\I\,\Diamond^\J\,\Phi)$$
		\Require $\rho(t)$ is the dynamics of a QCTMC $\QC$ and
			$\Box^\I\,\Diamond^\J\,\Phi$ is the repeated reachability in finite horizon to be checked;
		\Ensure $t_1^*,\ldots,t_m^*$ are finitely many absolute times, satisfying
			\[
				\bigwedge_{i=1}^m \rho(t_i^*) \models \Phi
				\quad \textup{and} \quad
				\I \subseteq \bigcup_{i=1}^m [t_i^*-\sup\J, t_i^*-\inf\J],
			\]
			whenever $\rho(0) \models \Box^\I\,\Diamond^\J\,\Phi$.
		\State let $\phi(t)$ be the observing expression of $\Phi$ as defined in Eq.~\eqref{eq:EP1};
		\State $\BB \gets [\inf\I+\inf\J,\sup\I+\sup\J]$
		that is the post-monitoring period of $\Box^\I\,\Diamond^\J\,\Phi$;
		\State $\I' \gets \emptyset$ and $\mathbb{T} \gets \emptyset$;
		\While{$\BB \ne \emptyset$ and $\I\setminus\I' \ne \emptyset$}
			\State let $t^*$ be an element of $\BB$;
			\State compute the sign-invariant neighborhood $\delta$ of $t^*$ by Eq.~\eqref{eq:neighbor};
			\If{$\rho(t^*) \models \Phi$}
                \If{$|\delta|\le|\J|$} $\I' \gets \I' \cup [t^*-\sup\J,t^*-\inf\J]$
				    and $\mathbb{T} \gets \mathbb{T} \cup \{t^*\}$;
                \Else 
				    \State let $l$ be an element in $(\inf\delta,\inf\delta+|\J|/2]$;
				    \Comment{set the leftmost sample}
				    \State let $u$ be an element in $[\sup\delta-|\J|/2,\sup\delta)$;
				    \Comment{set the rightmost sample}
                    \State let $l=t_1^*<t_2^*<\cdots<t_k^*=u$ be
				        the shortest arithmetic progression with common difference $\le |\J|$;
				        \Comment{set the intermediate samples}
				    \State $\I' \gets \I' \cup [t_1^*-\sup\J,t_k^*-\inf\J]$
				        and $\mathbb{T} \gets \mathbb{T} \cup \{t_1^*,t_2^*,\ldots,t_k^*\}$;
                \EndIf
			\EndIf
			\State $\BB \gets \BB \setminus \delta$;
		\EndWhile
		\If{$\I\setminus\I' = \emptyset$} \Return $\mathbb{T}$;
		\Else\ \Return $\emptyset$ as reporting $\rho(0) \not\models \Box^\I\,\Diamond^\J\,\Phi$.
		\EndIf
	\end{algorithmic}
\end{algorithm}

\begin{example}\label{ex7}
Reconsidering Example~\ref{ex6},
we now are to decide whether the initial ID $\rho(0)$ meets
the repeated reachability $\Psi\equiv\Box^{\mathcal{K}}\,\Diamond^\J\,\Phi$,
where $\mathcal{K}=[1,2]$ is a fresh time interval.
The post-monitoring period $\mnt(\Psi)$ is $\sup\mathcal{K}+\sup\J=3 $ by Eq.~\eqref{eq:EP6},
implying $\BB=[1,3]$.
Using Algorithm~\ref{solve},
we get a finite union of sign-invariant neighborhoods $\bigcup_{i=5}^{16} \delta_i \approx(0.852464, 3.26579)$,
which can cover the whole sample space $\BB$
but contribute the coverage $\I'=[0, \tfrac{3}{2}]$ partially covering the target $\mathcal{K}$.
More details can be found in Table~\ref{table_ex7}.
Hence, the repeated reachability $\Box^{\mathcal{K}}\,\Diamond^\J\,\Phi$ does not hold at the initial ID $\rho(0)$,
i.\,e.\@ $\rho(0) \not\models \Box^{\mathcal{K}}\,\Diamond^\J\,\Phi$.

\begin{table}[ht!]
	\caption{Sample-driven procedure for deciding $\Box^{\mathcal{K}}\,\Diamond^\J\,\Phi$}\label{table_ex7}
	\centering
	{\small
	\begin{tabular}{p{2cm}<{\centering}cccp{2cm}<{\centering}}
	\toprule
	samples $t^*$ & radius & neighborhood $\delta$ & ess. samples $\mathbb{T}$ & coverage $\I'$ \\
	\midrule
	$t_5^*=3$ & \multirow{2}{*}{$\epsilon_5=\tfrac{332241}{1250000}$}
	& $\delta_5=(t_5^*-\epsilon_5,t_5^*+\epsilon_5)$ & \multirow{2}{*}{$\emptyset$} & \\
	conflicting && $\approx (2.73421, 3.26579)$ \\
	\midrule
	$t_6^*=\tfrac{5}{2}$ & \multirow{2}{*}{$\epsilon_6=\tfrac{328869}{1250000}$}
	& $\delta_6=(t_6^*-\epsilon_6,t_6^*+\epsilon_6)$ & \multirow{2}{*}{$\emptyset$} & \\
	conflicting && $\approx (2.23691, 2.76309)$ & \\
	\midrule
	$t_7^*=\tfrac{11}{5}$ & \multirow{2}{*}{$\epsilon_7=\frac{286513}{1250000}$}
	& $\delta_7=(t_7^*-\epsilon_7,t_7^*+\epsilon_7)$ & \multirow{2}{*}{$\emptyset$} & \\
	conflicting && $\approx (1.97079, 2.42921)$ \\
	\midrule
	$t_8^*=\tfrac{19}{10}$ & \multirow{2}{*}{$\epsilon_8=\frac{738173}{5000000}$}
	& $\delta_8=(t_8^*-\epsilon_8,t_8^*+\epsilon_8)$ & \multirow{2}{*}{$\emptyset$} & \\
	conflicting && $\approx (1.75237, 2.04763)$ \\
	\midrule
	$t_9^*=\tfrac{17}{10}$ & \multirow{2}{*}{$\theta_9=\frac{33937}{500000}$}
	& $\delta_9=(t_9^*-\theta_9,t_9^*+\theta_9)$ & \multirow{2}{*}{$\emptyset$} & \\
	conflicting && $\approx (1.63213, 1.76787)$ \\
	\midrule
	$t_{10}^*=\tfrac{8}{5}$ & \multirow{2}{*}{$\theta_{10}=\tfrac{21821}{312500}$}
	& $\delta_{10}=(s_2^*, t_{10}^*+\theta_{10})$ & \multirow{2}{*}{$\emptyset$} & \\
	conflicting && $\approx (1.56093, 1.66982)$ \\
	\midrule
	$t_{11}^*=s_2^*$ & \multirow{2}{*}{$0$} & $\delta_{11}=\{s_2^*\}$ & \multirow{2}{*}{$\emptyset$} & \\
	conflicting && $\approx \{1.56093\}$ \\
	\midrule
	$t_{12}^*=\tfrac{3}{2}$ & \multirow{2}{*}{$\theta_{12}=\tfrac{25763}{390625}$}
	& $\delta_{12}=(t_{12}^*-\theta_{12}, s_2^*)$ & \multirow{2}{*}{$\{\tfrac{3}{2}\}$}
	& \multirow{2}{*}{$[\tfrac{1}{2},\tfrac{3}{2}]$} \\
	satisfying && $\approx(1.43405,1.56093)$ \\
	\midrule
	$t_{13}^*=\tfrac{7}{5}$ & \multirow{2}{*}{$\epsilon_{13}=\tfrac{342521}{5000000}$}
	& $\delta_{13}=(t_{13}^*-\epsilon_{13},t_{13}^*+\epsilon_{13})$ & \multirow{2}{*}{$\{\tfrac{7}{5}\}$}
	& \multirow{2}{*}{$[\tfrac{2}{5},\tfrac{7}{5}]$} \\
	satisfying && $\approx (1.33150,1.46850)$ \\
	\midrule
	$t_{14}^*=\tfrac{13}{10}$ & \multirow{2}{*}{$\epsilon_{14}=\tfrac{4879431}{50000000}$} 
	& $\delta_{14}=(t_{14}^*-\epsilon_{14},t_{14}^*+\epsilon_{14})$
	& \multirow{2}{*}{$\{\tfrac{13}{10}\}$} & \multirow{2}{*}{$[\tfrac{3}{10},\tfrac{13}{10}]$} \\
	satisfying && $\approx (1.20242, 1.39758)$ \\
	\midrule
	$t_{15}^*=\tfrac{6}{5}$ & \multirow{2}{*}{$\epsilon_{15}=\tfrac{132921}{1250000}$}
	& $\delta_{15}=(t_{15}^*-\epsilon_{15},t_{15}^*+\epsilon_{15})$
	& \multirow{2}{*}{$\{\tfrac{6}{5}\}$} & \multirow{2}{*}{$[\tfrac{1}{5},\tfrac{6}{5}]$} \\
	satisfying && $\approx(1.09367,1.30633)$ \\
	\midrule
	$t_{16}^*=1$ & \multirow{2}{*}{$\theta_{16}=\tfrac{3722353}{25000000}$}
	& $\delta_{16}=(s_{16}^*, t_{16}^*+\theta_{16})$ & \multirow{2}{*}{$\{1\}$} & \multirow{2}{*}{$[0,1]$} \\
	satisfying && $\approx (0.852464, 1.14889)$ \\
	\bottomrule
	\end{tabular}}
	\begin{tablenotes}
	\item Here, $s_2^*=\lambda_2 \approx 1.56093$ (see Fig.~\ref{fig1}),
	which is the unique real root of $\phi(t)$ during the interval $(t_{10}^*-\theta_{10},t_{10}^*)$.
	\end{tablenotes}
	\end{table}
\end{example}

\section{Experimentation}\label{S7}
The prototypes of both the presented sample-driven solving procedure (Algorithm~\ref{solve})
and the isolation-based one in the previous work~\cite[Algorithm~1]{XMGY21} have been implemented in Python~3.8,
running on an Apple M1 core with 16 GB memory.
We have experimented on randomly generated instances from the two-qubit Hilbert space.
Specifically speaking,
symbolically computing the exponentials of high-dimensional matrices and manipulating their entries
is well known to be of expensive computational cost,
but is out of what we mainly concern in the present paper.
So we fix the sample dynamics of the QCTMC $\QC_1$ shown in Example~\ref{ex1}
to demonstrate the performance of two procedures.
Whereas, the randomness comes from two aspects:
\begin{enumerate}
\item the observing expressions are randomly generated of various degrees and heights
(that is
the maximum of absolute values of a $\mathbb{Z}$-polynomial's coefficients,
which is usually used to reflect the size of that polynomial;
e.\,g.\@, the height of the polynomial $6x_2-2x_1-9x_4-3$ is $\max\{|6|,|-2|,|-9|,|-3|\}=9$),
\item the time intervals $\I$ and $\J$ are also randomly generated.
\end{enumerate}

The effectiveness, efficiency and scalability of the two procedures are validated by
the time and space consumption on randomly generated signals $\Phi$,
in which observing expressions $p(\x)$ are measured in terms of
i) the degree of an integer polynomial $p(\x)$ chosen from $1$ to $4$,
and ii) the height of $p(\x)$ selected at four scales: $[1,10]$, $[11,100]$, $[101,500]$ and $[501,1000]$.
The experimental results are summarized in the middle column of Table~\ref{table:poly_exp}.
We also generalize the inner single signal $\Phi$
by multiple ones composed in conjunctive normal form (CNF),
saying $\Phi \equiv (\Phi_1 \vee \Phi_2 \vee \Phi_3) \wedge (\Phi_4 \vee \Phi_5\vee \Phi_6)$
investigated in the experiments,
where
\[
	\begin{array}{lll}
		\Phi_1 \equiv 6x_2-2x_1-9x_4-3 > 0 & &\quad \Phi_4 \equiv 4x_1+3x_2-8x_3+4 > 0 \\
		\Phi_2 \equiv 4x_1+3x_2-8x_3+4 > 0 & &\quad \Phi_5 \equiv 2x_3+6 < 0 \\
		\Phi_3 \equiv 4x_2+9 \le 0 & &\quad \Phi_6 \equiv 2x_1+6 < 0.
	\end{array}
\]
Supposing that $\delta_i$ is a neighborhood on which $\Phi_i$ is met everywhere,
$(\delta_1 \cup \delta_2 \cup \delta_3) \cap (\delta_4 \cup \delta_5 \cup \delta_6)$
is a truth-invariant neighborhood of $\Phi$;
supposing that $\delta_i$ is a neighborhood on which $\Phi_i$ is met nowhere,
$(\delta_1 \cap \delta_2 \cap \delta_3) \cup (\delta_4 \cap \delta_5 \cap \delta_6)$
is a truth-invariant neighborhood of $\Phi$.
Finally the results are summarized in the right column of Table~\ref{table:poly_exp},
while the source code and complete experimental data can be found at
\href{https://github.com/Holly-Jiang/RR.git}{\textcolor{black}{https://github.com/Holly-Jiang/RR}}.

\begin{table}[ht]
    \caption{Performance of the isolation-based and the sample-driven solving procedures}\label{table:poly_exp}
    \centering
    \belowrulesep=0pt
\aboverulesep=0pt
	\begin{tabular}{cc|cccc|cccc}
	\toprule
	\multirow{3}{*}{degree} & \multirow{3}{*}{height}
	& \multicolumn{4}{c|}{single signals} & \multicolumn{4}{c}{multiple signals} \\
    & & \multicolumn{2}{c}{isolation-based} & \multicolumn{2}{c|}{sample-driven}
	& \multicolumn{2}{c}{isolation-based} & \multicolumn{2}{c}{sample-driven} \\
	& & time (s) & space (MB) & time & space & time & space & time & space \\
	\midrule
	\multirow{4}{*}{1} & $[1,10]$ & 9.55 & 121 & \textbf{2.51} & \textbf{113}
	& 31.98 & 465 & \textbf{16.72} & \textbf{459} \\
	& $[11,100]$ & 7.12 & 117 & \textbf{1.80} & \textbf{108} 
        & 21.69 & 460 & \textbf{ 9.19} & \textbf{422} \\
	& $[101, 500]$ & 9.21 & 120 & \textbf{1.95} & \textbf{109} & 27.90 & 469 & \textbf{10.77} & \textbf{425} \\
	& $[501,1000]$ &  6.02 & 117 & \textbf{ 1.81} & \textbf{108} & 23.96 & 370 & \textbf{7.65} & \textbf{333} \\
	\midrule
	\multirow{4}{*}{2} & $[1,10]$ & 5.43 & 115 & \textbf{1.83} & \textbf{107}
	&  33.85 & 432 & \textbf{12.45} & \textbf{416} \\
	& $[11,100]$ & 7.77 & 117 & \textbf{1.92} & \textbf{109} & 20.47 & 268 & \textbf{10.89} & \textbf{255} \\
	& $[101,500]$ & 15.19 & 119 & \textbf{4.30} &\textbf{110} & 27.45 & 392 & \textbf{16.87} & \textbf{377} \\
    & $[501,1000]$ & 8.48 & 117 & \textbf{1.96} & \textbf{109} & 42.96 & 412 & \textbf{21.31} & \textbf{402} \\
	\midrule
	\multirow{4}{*}{3} & $[1,10]$ & 27.12 & 127 & \textbf{7.61} & \textbf{117}
	& 104.47 & 411 & \textbf{46.33} & \textbf{389} \\
	& $[11,100]$ & 46.73 & 125 & \textbf{4.75} & \textbf{115} & 48.71 & 300 & \textbf{18.09} & \textbf{293} \\
	& $[101,500]$ & 43.79 & 131 & \textbf{2.17} & \textbf{110}& 80.18 & 411 & \textbf{24.14} & \textbf{391} \\
    & $[501,1000]$ & 11.12& 121 & \textbf{2.32} & \textbf{110} &128.51 & 495 & \textbf{46.97} & \textbf{468} \\
	\midrule
	\multirow{4}{*}{4} & $[1,10]$ & 24.81 & 128 & \textbf{2.31} & \textbf{110}
	& 239.84 & 616 & \textbf{103.07} & \textbf{525} \\
	& $[11,100]$ & 31.36 & 127 & \textbf{2.17} & \textbf{110} & 76.67 & 408 & \textbf{ 16.51} & \textbf{392} \\
	& $[101,500]$ &  40.45 & 129 & \textbf{2.54} & \textbf{112} & 102.20 & 302 & \textbf{32.78} & \textbf{271} \\
    & $[501,1000]$ & 21.72 & 126  & \textbf{11.07} & \textbf{120} &  105.14 & 290 & \textbf{65.81} & \textbf{253} \\
	\bottomrule
	\end{tabular}
\end{table}

\begin{figure}[htp]
	\centering
	\scalebox{0.75}{
		\begin{tikzpicture}
			\begin{axis}[
				small,
				axis y line*=left,
				bar width=5pt,
				enlargelimits=0,
				axis y line*=left,
				legend style={at={(0.5,-0.15)},
					anchor=north,legend columns=-1},
				ylabel={time consumption (s)},
				xtick={0,5,10,15,20},
				xticklabels={$ $, $\leq 10$, $\leq 100$,$\leq 500$, $\leq 1000$},
				y dir=normal,
				legend columns=2,
				legend style={draw=none},
				ymin=0, ymax=30
				]
				\addplot[nodes near coords, color=red1, mark=*,point meta=explicit symbolic,scatter/classes={
					a={mark=square*,red1},
					c={mark=square,draw=red1}}] table [x=x,y=a,meta=sample] {./data/data_d1.txt};
				\addplot[nodes near coords, color=blue2, mark=*,point meta=explicit symbolic,scatter/classes={
					a={mark=square*,blue2},
					c={mark=square,draw=blue2}}] table [x=x,y=c,meta=iso] {./data/data_d1.txt};
				\legend{,sample-driven,, isolation-based}
			\end{axis}
			\begin{axis}[
				small,
				axis y line*=right,
				axis x line=none,
				enlargelimits=0,
				legend style={at={(0.5,-0.15)},
					anchor=north,legend columns=-1},
				ylabel={space consumption (MB)},
				xtick={0,5,10,15,20},
				xticklabels={$ $, $\leq 10$, $\leq 100$,$\leq 500$, $\leq 1000$},
				y dir=normal,
				ymin=80, ymax=130
				]
				\addplot [nodes near coords, color=blue2,only marks, mark=*,point meta=explicit symbolic,scatter/classes={
					a={mark=square*,blue2},
					c={mark=square,draw=blue2}}] table [ x=x,y=d,meta=iso] {./data/data_d1.txt};
				\addplot [nodes near coords,only marks,  color=red1,scatter, mark=*,point meta=explicit symbolic,		scatter/classes={
					a={mark=square*,red1},
					c={mark=square,draw=red1}}] table [ x=x,y=b,meta=sample] {./data/data_d1.txt};
	\draw [blue2, dotted](1,105)--(1,111);
	\draw [blue2, dotted](2,105)--(2,111);
	\draw [blue2, dotted](3,106)--(3,113);
	\draw [blue2, dotted](4,106)--(4,113);
	\draw [blue2, dotted](5,106)--(5,113);
	\draw [blue2, dotted](6,106)--(6,106);
	\draw [blue2, dotted](7,110)--(7,115);
	\draw [blue2, dotted](8,110)--(8,116);
	\draw [blue2, dotted](9,110)--(9,116);
	\draw [blue2, dotted](10,110)--(10,116);
	\draw [blue2, dotted](11,101)--(11,111);
	\draw [blue2, dotted](12,102)--(12,114);
	\draw [blue2, dotted](13,102)--(13,116);
	\draw [blue2, dotted](14,106)--(14,118);
	\draw [blue2, dotted](15,106)--(15,118);
	\draw [blue2, dotted](16,102)--(16,111);
	\draw [blue2, dotted](17,103)--(17,113);
	\draw [blue2, dotted](18,104)--(18,115);
	\draw [blue2, dotted](19,104)--(19,115);
	\draw [blue2, dotted](20,105)--(20,116);
			\end{axis}
		\node at(2.4,-1.5){(a) single signal, degree $=1$};
		\end{tikzpicture}
		\begin{tikzpicture}
			\begin{axis}[
				small,
				enlargelimits=0,
				axis y line*=left,
				legend style={at={(0.5,-0.15)},
					anchor=north,legend columns=-1},
				ylabel={time consumption (s)},
				xtick={0,5,10,15,20},
				xticklabels={$ $, $\leq 10$, $\leq 100$,$\leq 500$, $\leq 1000$},
				y dir=normal,
				legend columns=2,
				legend style={draw=none},
				ymin=0, ymax=100
				]
				\addplot[nodes near coords,color=red1, mark=square*, point meta=explicit symbolic,scatter/classes={
					a={mark=square*,red1},
					c={mark=square,draw=red1}}] table [x=x,y=a,meta=sample] {./data/data_d2.txt};\label{time_data2_sample};
				\addplot[nodes near coords,color=blue2, mark=square*,point meta=explicit symbolic,scatter/classes={
					a={mark=square*,blue2},
					c={mark=square,draw=blue2}}] table [x=x,y=c,meta=iso] {./data/data_d2.txt};\label{time_data2_isolate};
				\legend{,sample-driven,, isolation-based}
			\end{axis};
			\begin{axis}[
				axis y line*=right,
				axis x line=none,
				small,
				enlargelimits=0,
				legend style={at={(0.5,-0.15)},
					anchor=north,legend columns=-1},
				ylabel={space consumption (MB)},
				y dir =normal,
				ymin=80, ymax=140
				]
				\addplot [nodes near coords,color=blue2,only marks, mark=*,point meta=explicit symbolic,
				scatter/classes={
					a={mark=square*,blue2},
					c={mark=square,draw=blue2}}] table [ x=x,y=d,meta=iso] {./data/data_d2.txt};
				\addplot [nodes near coords,only marks ,color=red1,scatter, mark=*,point meta=explicit symbolic,
				scatter/classes={
					a={mark=square*,red1},
					c={mark=square,draw=red1}}] table [ x=x,y=b,meta=sample] {./data/data_d2.txt};
	\draw [blue2, dotted](1,106)--(1,120);
	\draw [blue2, dotted](2,107)--(2,121);
	\draw [blue2, dotted](3,107)--(3,121);
	\draw [blue2, dotted](4,107)--(4,121);
	\draw [blue2, dotted](5,107)--(5,121);
	\draw [blue2, dotted](6,108)--(6,110);
	\draw [blue2, dotted](7,108)--(7,113);
	\draw [blue2, dotted](8,109)--(8,115);
	\draw [blue2, dotted](9,109)--(9,116);
	\draw [blue2, dotted](10,109)--(10,117);
	\draw [blue2, dotted](11,112)--(11,114);
	\draw [blue2, dotted](12,112)--(12,116);
	\draw [blue2, dotted](13,112)--(13,116);
	\draw [blue2, dotted](14,112)--(14,117);
	\draw [blue2, dotted](15,112)--(15,119);
	\draw [blue2, dotted](16,104)--(16,110);
	\draw [blue2, dotted](17,107)--(17,117);
	\draw [blue2, dotted](18,107)--(18,117);
	\draw [blue2, dotted](19,107)--(19,117);
	\draw [blue2, dotted](20,108)--(20,118);
			\end{axis}
		\node at(2.4,-1.5){(b) single signal, degree $=2$};
		\end{tikzpicture}
	}
	\scalebox{0.75}{
		\begin{tikzpicture}
			\begin{axis}[
				small,
				axis y line*=left,
				bar width=5pt,
				enlargelimits=0,
				axis y line*=left,
				legend style={at={(0.5,-0.15)},
					anchor=north,legend columns=-1},
				ylabel={time consumption (s)},
				xtick={0,5,10,15,20},
				xticklabels={$ $, $\leq 10$, $\leq 100$,$\leq 500$, $\leq 1000$},
				y dir=normal,
				legend columns=2,
				legend style={draw=none},
				ymin=0, ymax=200
				]
				\addplot[nodes near coords, color=red1, mark=*,point meta=explicit symbolic,		scatter/classes={
					a={mark=square*,red1},
					c={mark=square,draw=red1}}] table [x=x,y=a,meta=sample] {./data/data_d3.txt};
				\addplot[nodes near coords, color=blue2, mark=*,point meta=explicit symbolic,		scatter/classes={
					a={mark=square*,blue2},
					c={mark=square,draw=blue2}}] table [x=x,y=c,meta=iso] {./data/data_d3.txt};
				\legend{,sample-driven,, isolation-based}
			\end{axis}
			\begin{axis}[
				small,
				axis y line*=right,
				axis x line=none,
				enlargelimits=0,
				legend style={at={(0.5,-0.15)},
					anchor=north,legend columns=-1},
				ylabel={space consumption (MB)},
				xtick={0,5,10,15,20},
				xticklabels={$ $, $\leq 10$, $\leq 100$,$\leq 500$, $\leq 1000$},
				y dir=normal,
				ymin=80, ymax=130
				]
				\addplot [nodes near coords, color=blue2,only marks, mark=*,point meta=explicit symbolic,		scatter/classes={
					a={mark=square*,blue2},
					c={mark=square,draw=blue2}}] table [ x=x,y=d,meta=iso] {./data/data_d3.txt};
				\addplot [nodes near coords,only marks,  color=red1,scatter, mark=*,point meta=explicit symbolic,		scatter/classes={
					a={mark=square*,red1},
					c={mark=square,draw=red1}}] table [ x=x,y=b,meta=sample] {./data/data_d3.txt};
	\draw [blue2, dotted](1,104)--(1,123);
	\draw [blue2, dotted](2,104)--(2,124);
	\draw [blue2, dotted](3,105)--(3,124);
	\draw [blue2, dotted](4,106)--(4,125);
	\draw [blue2, dotted](5,106)--(5,125);
	\draw [blue2, dotted](6,101)--(6,116);
	\draw [blue2, dotted](7,101)--(7,116);
	\draw [blue2, dotted](8,101)--(8,118);
	\draw [blue2, dotted](9,101)--(9,119);
	\draw [blue2, dotted](10,101)--(10,120);
	\draw [blue2, dotted](11,109)--(11,117);
	\draw [blue2, dotted](12,110)--(12,119);
	\draw [blue2, dotted](13,110)--(13,120);
	\draw [blue2, dotted](14,110)--(14,123);
	\draw [blue2, dotted](15,110)--(15,123);
	\draw [blue2, dotted](16,104)--(16,121);
	\draw [blue2, dotted](17,105)--(17,123);
	\draw [blue2, dotted](18,106)--(18,124);
	\draw [blue2, dotted](19,106)--(19,125);
	\draw [blue2, dotted](20,107)--(20,125);
			\end{axis}
		\node at(2.4,-1.5){(c) single signal, degree $=3$};
		\end{tikzpicture}
		\begin{tikzpicture}
			\begin{axis}[
				small,
				axis y line*=left,
				bar width=5pt,
				enlargelimits=0,
				axis y line*=left,
				legend style={at={(0.5,-0.15)},
					anchor=north,legend columns=-1},
				ylabel={time consumption (s)},
				xtick={0,5,10,15,20},
				xticklabels={$ $, $\leq 10$, $\leq 100$,$\leq 500$, $\leq 1000$},
				y dir=normal,
				legend columns=2,
				legend style={draw=none},
				ymin=0, ymax=200
				]
				\addplot[nodes near coords, color=red1, mark=*,point meta=explicit symbolic,		scatter/classes={
					a={mark=square*,red1},
					c={mark=square,draw=red1}}] table [x=x,y=a,meta=sample] {./data/data_d4.txt};
				\addplot[nodes near coords, color=blue2, mark=*,point meta=explicit symbolic,		scatter/classes={
					a={mark=square*,blue2},
					c={mark=square,draw=blue2}}] table [x=x,y=c,meta=iso] {./data/data_d4.txt};
				\legend{,sample-driven,, isolation-based}
			\end{axis}
			\begin{axis}[
				small,
				axis y line*=right,
				axis x line=none,
				enlargelimits=0,
				legend style={at={(0.5,-0.15)},
					anchor=north,legend columns=-1},
				ylabel={space consumption (MB)},
				xtick={0,5,10,15,20},
				xticklabels={$ $, $\leq 10$, $\leq 100$,$\leq 500$, $\leq 1000$},
				y dir=normal,
				ymin=80, ymax=135
				]
				\addplot [nodes near coords, color=blue2,only marks, mark=*,point meta=explicit symbolic,		scatter/classes={
					a={mark=square*,blue2},
					c={mark=square,draw=blue2}}] table [ x=x,y=d,meta=iso] {./data/data_d4.txt};
				\addplot [nodes near coords,only marks,  color=red1,scatter, mark=*,point meta=explicit symbolic,		scatter/classes={
					a={mark=square*,red1},
					c={mark=square,draw=red1}}] table [ x=x,y=b,meta=sample] {./data/data_d4.txt};
	\draw [blue2, dotted](1,107.0)--(1,106.0);
	\draw [blue2, dotted](2,119.0)--(2,118.0);
	\draw [blue2, dotted](3,120.0)--(3,127.0);
	\draw [blue2, dotted](4,120.0)--(4,127.0);
	\draw [blue2, dotted](5,120.0)--(5,127.0);
	\draw [blue2, dotted](6,119.0)--(6,113.0);
	\draw [blue2, dotted](7,119.0)--(7,128.0);
	\draw [blue2, dotted](8,119.0)--(8,129.0);
	\draw [blue2, dotted](9,119.0)--(9,129.0);
	\draw [blue2, dotted](10,120.0)--(10,131.0);
	\draw [blue2, dotted](11,111.0)--(11,126.0);
	\draw [blue2, dotted](12,112.0)--(12,126.0);
	\draw [blue2, dotted](13,112.0)--(13,127.0);
	\draw [blue2, dotted](14,113.0)--(14,129.0);
	\draw [blue2, dotted](15,113.0)--(15,129.0);
	\draw [blue2, dotted](16,110.0)--(16,112.0);
	\draw [blue2, dotted](17,110.0)--(17,115.0);
	\draw [blue2, dotted](18,110.0)--(18,118.0);
	\draw [blue2, dotted](19,111.0)--(19,120.0);
	\draw [blue2, dotted](20,111.0)--(20,124.0);
			\end{axis}
		\node at(2.4,-1.5){(d) single signal, degree $=4$};
		\end{tikzpicture}
	}
	\scalebox{0.75}{
		\begin{tikzpicture}
			\begin{axis}[
				small,
				axis y line*=left,
				enlargelimits=0,
				legend style={at={(0.5,-0.15)},
					anchor=north,legend columns=-1},
				ylabel={time consumption (s)},
				xtick={0,5,10,15,20},
				xticklabels={$ $, $\leq 10$, $\leq 100$,$\leq 500$, $\leq 1000$},
				y dir=normal,
				legend columns=2,
				legend style={draw=none},
				ymin=0, ymax=200
				]
				\addplot[nodes near coords, color=red1, mark=*,point meta=explicit symbolic,		scatter/classes={
					a={mark=square*,red1},
					c={mark=square,draw=red1}}] table [x=x,y=a, meta=sample] {./data/data_bool_d1.txt};
				\addplot[nodes near coords, color=blue2, mark=*,point meta=explicit symbolic,		scatter/classes={
					a={mark=square*,blue2},
					c={mark=square,draw=blue2}}] table [x=x,y=c, meta=iso] {./data/data_bool_d1.txt};
				\legend{,sample-driven, ,isolation-based}
			\end{axis}
			\begin{axis}[
				small,
				axis y line*=right,
				axis x line=none,
				enlargelimits=0,
				legend style={at={(0.5,-0.15)},
					anchor=north,legend columns=-1},
				ylabel={space consumption (MB)},
				xtick={0,5,10,15,20},
				xticklabels={$ $, $\leq 10$, $\leq 100$,$\leq 500$, $\leq 1000$},
				y dir=normal,
				ymin=0, ymax=1800
				]
				\addplot [nodes near coords,color=blue2, only marks, mark=*,point meta=explicit symbolic,
				scatter/classes={
					a={mark=square*,blue2},
					c={mark=square,draw=blue2}}] table [ x=x,y=d,meta=iso] {./data/data_bool_d1.txt};
				\addplot [nodes near coords,only marks,  color=red1,scatter, mark=*,point meta=explicit symbolic,
				scatter/classes={
					a={mark=square*,red1},
					c={mark=square,draw=red1}}] table [ x=x,y=b,meta=sample] {./data/data_bool_d1.txt};
	\draw [blue2, dotted](1,528)--(1,588);
	\draw [blue2, dotted](2,324)--(2,358);
	\draw [blue2, dotted](3,1081)--(3,1195);
	\draw [blue2, dotted](4,866)--(4,965);
	\draw [blue2, dotted](5,759)--(5,848);
	\draw [blue2, dotted](6,829)--(6,902);
	\draw [blue2, dotted](7,840)--(7,935);
	\draw [blue2, dotted](8,1602)--(8,1775);
	\draw [blue2, dotted](9,857)--(9,950);
	\draw [blue2, dotted](10,323)--(10,358);
	\draw [blue2, dotted](11,432)--(11,449);
	\draw [blue2, dotted](12,656)--(12,715);
	\draw [blue2, dotted](13,441)--(13,484);
	\draw [blue2, dotted](14,663)--(14,730);
	\draw [blue2, dotted](15,885)--(15,975);
	\draw [blue2, dotted](16,1042)--(16,1138);
	\draw [blue2, dotted](17,963)--(17,1066);
	\draw [blue2, dotted](18,650)--(18,716);
	\draw [blue2, dotted](19,434)--(19,478);
	\draw [blue2, dotted](20,759)--(20,837);
			\end{axis}
	  \node at(2.4,-1.5){(e) CNF of multiple signals, degree $=1$};
		\end{tikzpicture}
			\begin{tikzpicture}
			\begin{axis}[
				small,
				bar width=5pt,
				axis y line*=left,
				enlargelimits=0,
				legend style={at={(0.5,-0.15)},
					anchor=north,legend columns=-1},
				ylabel={time consumption (s)},
				xtick={0,5,10,15,20},
				xticklabels={$ $, $\leq 10$, $\leq 100$,$\leq 500$, $\leq 1000$},
				y dir=normal,
				ymin=0, ymax=200,
				legend columns=2,
				legend style={draw=none}
				]
				\addplot [nodes near coords,color=red1, mark=*,point meta=explicit symbolic,
				scatter/classes={
					a={mark=square*,red1},
					c={mark=square,draw=red1}}] table [x=x,y=a, meta=sample] {./data/data_bool_d2.txt};\label{pgfplots:label1}
				\addplot [nodes near coords,color=blue2, mark=*,point meta=explicit symbolic,
				scatter/classes={
					a={mark=square*,blue2},
					c={mark=square,draw=blue2}}] table [x=x,y=c, meta=iso] {./data/data_bool_d2.txt};\label{pgfplots:label2}
				\legend{,sample-driven, , isolation-based}
			\end{axis}
			\begin{axis}[
				small,
				bar width=5pt,
				enlargelimits=0,
				axis y line*=right,
				axis x line=none,
				legend style={at={(0.5,-0.15)},
					anchor=north,legend columns=-1},
				ylabel={space consumption (MB)},
				xtick={0,5,10,15,20},
				xticklabels={$ $, $\leq 10$, $\leq 100$,$\leq 500$, $\leq 1000$},
				y dir=normal,
				ymin=0, ymax=1400
				]
				\addplot [nodes near coords,only marks, color=blue2,scatter, mark=*,point meta=explicit symbolic,
				scatter/classes={
					a={mark=square*,blue2},
					c={mark=square,draw=blue2}}] table [ x=x,y=d,meta=iso] {./data/data_bool_d2.txt};
				\addplot [nodes near coords,only marks, color=red1,scatter, mark=*,point meta=explicit symbolic,
				scatter/classes={
					a={mark=square*,red1},
					c={mark=square,draw=red1}}] table [ x=x,y=b,meta=sample] {./data/data_bool_d2.txt};
	\draw [blue2, dotted](1,816)--(1,910);
	\draw [blue2, dotted](2,735)--(2,846);
	\draw [blue2, dotted](3,639)--(3,736);
	\draw [blue2, dotted](4,1298)--(4,1356);
	\draw [blue2, dotted](5,762)--(5,865);
	\draw [blue2, dotted](6,201)--(6,212);
	\draw [blue2, dotted](7,859)--(7,941);
	\draw [blue2, dotted](8,764)--(8,840);
	\draw [blue2, dotted](9,446)--(9,487);
	\draw [blue2, dotted](10,782)--(10,854);
	\draw [blue2, dotted](11,527)--(11,570);
	\draw [blue2, dotted](12,752)--(12,836);
	\draw [blue2, dotted](13,432)--(13,485);
	\draw [blue2, dotted](14,434)--(14,488);
	\draw [blue2, dotted](15,793)--(15,869);
	\draw [blue2, dotted](16,937)--(16,1025);
	\draw [blue2, dotted](17,437)--(17,483);
	\draw [blue2, dotted](18,331)--(18,338);
	\draw [blue2, dotted](19,887)--(19,920);
	\draw [blue2, dotted](20,667)--(20,695);
			\end{axis}
		\node at(2.4,-1.5){(f) CNF of multiple signals, degree $=2$};
		\end{tikzpicture}
	}
	\scalebox{0.75}{
		\begin{tikzpicture}
			\begin{axis}[
				small,
				axis y line*=left,
				bar width=5pt,
				enlargelimits=0,
				axis y line*=left,
				legend style={at={(0.5,-0.15)},
					anchor=north,legend columns=-1},
				ylabel={time consumption (s)},
				xtick={0,5,10,15,20},
				xticklabels={$ $, $\leq 10$, $\leq 100$,$\leq 500$, $\leq 1000$},
				y dir=normal,
				legend columns=2,
				legend style={draw=none},
				ymin=0, ymax=500
				]
				\addplot[nodes near coords, color=red1, mark=*,point meta=explicit symbolic,		scatter/classes={
					a={mark=square*,red1},
					c={mark=square,draw=red1}}] table [x=x,y=a,meta=sample] {./data/data_bool_d3.txt};
				\addplot[nodes near coords, color=blue2, mark=*,point meta=explicit symbolic,		scatter/classes={
					a={mark=square*,blue2},
					c={mark=square,draw=blue2}}] table [x=x,y=c,meta=iso] {./data/data_bool_d3.txt};
				\legend{,sample-driven,, isolation-based}
			\end{axis}
			\begin{axis}[
				small,
				axis y line*=right,
				axis x line=none,
				enlargelimits=0,
				legend style={at={(0.5,-0.15)},
					anchor=north,legend columns=-1},
				ylabel={space consumption (MB)},
				xtick={0,5,10,15,20},
				xticklabels={$ $, $\leq 10$, $\leq 100$,$\leq 500$, $\leq 1000$},
				y dir=normal,
				ymin=80, ymax=1600
				]
				\addplot [nodes near coords, color=blue2,only marks, mark=*,point meta=explicit symbolic,		scatter/classes={
					a={mark=square*,blue2},
					c={mark=square,draw=blue2}}] table [ x=x,y=d,meta=iso] {./data/data_bool_d3.txt};
				\addplot [nodes near coords,only marks,  color=red1,scatter, mark=*,point meta=explicit symbolic,		scatter/classes={
					a={mark=square*,red1},
					c={mark=square,draw=red1}}] table [ x=x,y=b,meta=sample] {./data/data_bool_d3.txt};
	\draw [blue2, dotted](1,544 )--(1,588.0);
	\draw [blue2, dotted](2,1022.0)--(2,1148.0);
	\draw [blue2, dotted](3,1042.0)--(3,1166.0);
	\draw [blue2, dotted](4,697.0)--(4,791.0);
	\draw [blue2, dotted](5,1053.0)--(5,1210.0);
	\draw [blue2, dotted](6,1216.0)--(6,1333.0);
	\draw [blue2, dotted](7,455.0)--(7,500.0);
	\draw [blue2, dotted](8,1376.0)--(8,1543.0);
	\draw [blue2, dotted](9,575.0)--(9,652.0);
	\draw [blue2, dotted](10,693.0)--(10,785.0);
	\draw [blue2, dotted](11,428.0)--(11,458.0);
	\draw [blue2, dotted](12,1011.0)--(12,1154.0);
	\draw [blue2, dotted](13,681.0)--(13,791.0);
	\draw [blue2, dotted](14,918.0)--(14,1052.0);
	\draw [blue2, dotted](15,462.0)--(15,495.0);
	\draw [blue2, dotted](16,1211.0)--(16,1328.0);
	\draw [blue2, dotted](17,227.0)--(17,233.0);
	\draw [blue2, dotted](18,910.0)--(18,938.0);
	\draw [blue2, dotted](19,456.0)--(19,471.0);
	\draw [blue2, dotted](20,689.0)--(20,709.0);
			\end{axis}
		\node at(2.4,-1.5){(g) CNF of multiple signals, degree $=3$};
		\end{tikzpicture}
		\begin{tikzpicture}
			\begin{axis}[
				small,
				axis y line*=left,
				bar width=5pt,
				enlargelimits=0,
				axis y line*=left,
				legend style={at={(0.5,-0.15)},
					anchor=north,legend columns=-1},
				ylabel={time consumption (s)},
				xtick={0,5,10,15,20},
				xticklabels={$ $, $\leq 10$, $\leq 100$,$\leq 500$, $\leq 1000$},
				y dir=normal,
				legend columns=2,
				legend style={draw=none},
				ymin=0, ymax=500
				]
				\addplot[nodes near coords, color=red1, mark=*,point meta=explicit symbolic,		scatter/classes={
					a={mark=square*,red1},
					c={mark=square,draw=red1}}] table [x=x,y=a,meta=sample] {./data/data_bool_d4.txt};
				\addplot[nodes near coords, color=blue2, mark=*,point meta=explicit symbolic,		scatter/classes={
					a={mark=square*,blue2},
					c={mark=square,draw=blue2}}] table [x=x,y=c,meta=iso] {./data/data_bool_d4.txt};
				\legend{,sample-driven,, isolation-based}
			\end{axis}
			\begin{axis}[
				small,
				axis y line*=right,
				axis x line=none,
				enlargelimits=0,
				legend style={at={(0.5,-0.15)},
					anchor=north,legend columns=-1},
				ylabel={space consumption (MB)},
				xtick={0,5,10,15,20},
				xticklabels={$ $, $\leq 10$, $\leq 100$,$\leq 500$, $\leq 1000$},
				y dir=normal,
				ymin=80, ymax=1800
				]
				\addplot [nodes near coords, color=blue2,only marks, mark=*,point meta=explicit symbolic,		scatter/classes={
					a={mark=square*,blue2},
					c={mark=square,draw=blue2}}] table [ x=x,y=d,meta=iso] {./data/data_bool_d4.txt};
				\addplot [nodes near coords,only marks,  color=red1,scatter, mark=*,point meta=explicit symbolic,		scatter/classes={
					a={mark=square*,red1},
					c={mark=square,draw=red1}}] table [ x=x,y=b,meta=sample] {./data/data_bool_d4.txt};
	\draw [blue2, dotted](1,534)--(1,569);
	\draw [blue2, dotted](2,1246)--(2,1385);
	\draw [blue2, dotted](3,1377)--(3,1450);
	\draw [blue2, dotted](4,1035)--(4,1103);
	\draw [blue2, dotted](5,576)--(5,615);
	\draw [blue2, dotted](6,1063)--(6,1232);
	\draw [blue2, dotted](7,684)--(7,769);
	\draw [blue2, dotted](8,687)--(8,787);
	\draw [blue2, dotted](9,576)--(9,662);
	\draw [blue2, dotted](10,1508)--(10,1726);
	\draw [blue2, dotted](11,1315)--(11,1539);
	\draw [blue2, dotted](12,341)--(12,399);
	\draw [blue2, dotted](13,800)--(13,938);
	\draw [blue2, dotted](14,1146)--(14,1346);
	\draw [blue2, dotted](15,460)--(15,541);
	\draw [blue2, dotted](16,318)--(16,374);
	\draw [blue2, dotted](17,750)--(17,883);
	\draw [blue2, dotted](18,868)--(18,1022);
	\draw [blue2, dotted](19,771)--(19,910);
	\draw [blue2, dotted](20,444)--(20,525);
			\end{axis}
		\node at(2.4,-1.5){(h) CNF of multiple signals, degree $=4$};
		\end{tikzpicture}
	}
\caption{Time and space consumption of the instances in Table~\ref{table:poly_exp}}\label{fig:exp}
\end{figure}
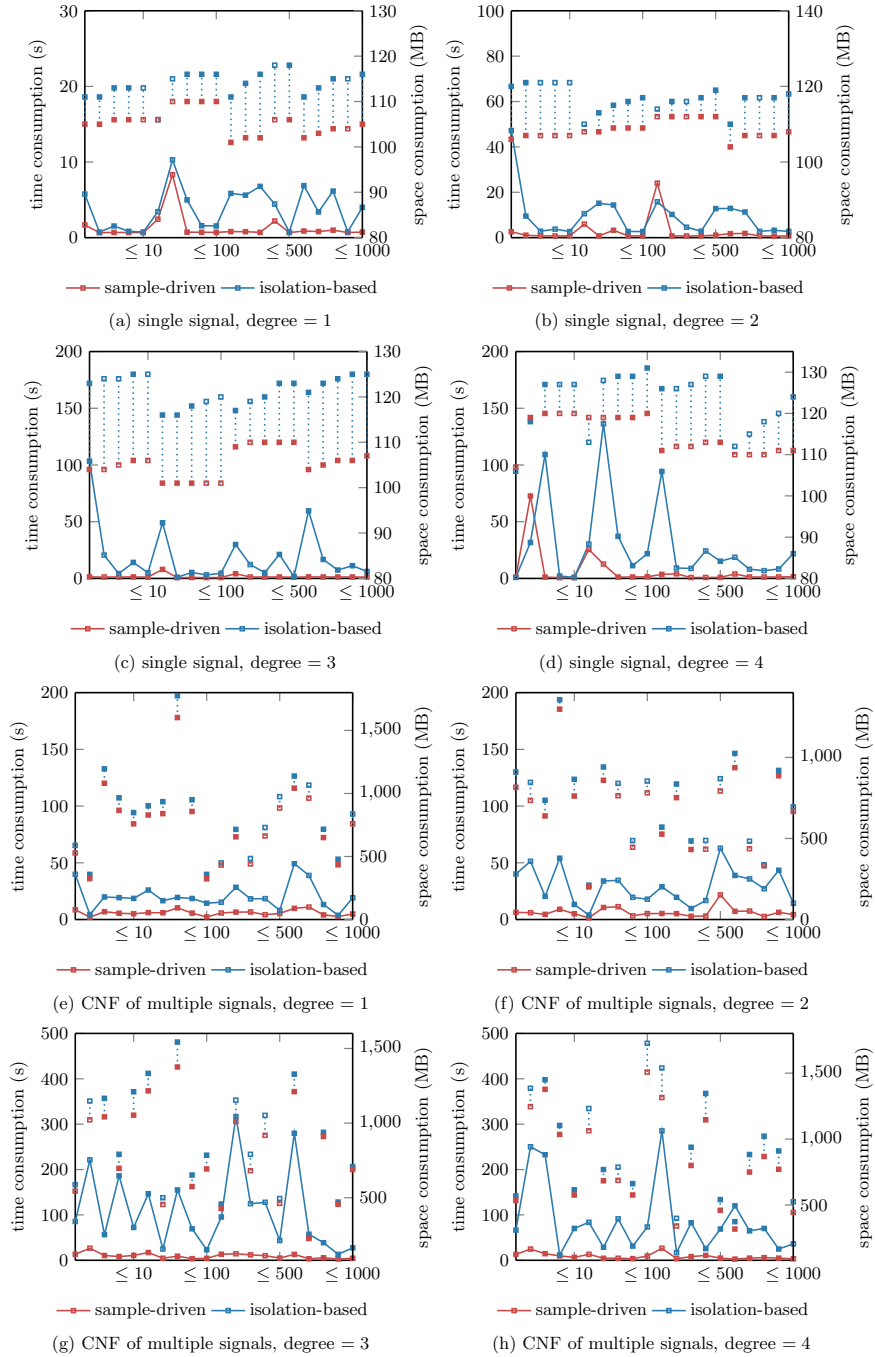

\paragraph{Effectiveness and efficiency of the two procedures}
For each row in Table~\ref{table:poly_exp},
we get the average time and space consumption of deciding five randomly generated instances of
a specified degree and height. 
In the case of observing expressions with a fixed degree,
there exists only a small fluctuation of time and space consumption,
as the height increases.
With the observing expressions changing from the single signal to the CNF of multiple signals,
the average time and space consumption increases.
Overall,
as demonstrated in Table~\ref{table:poly_exp} and Fig.~\ref{fig:exp},
the results on deciding the repeated reachability in our randomly generated instances
are achieved at an acceptable level of consumption
when applying two  procedures.

\paragraph{Superiority of sample-driven solving procedure over isolation-based one}
In Fig.~\ref{fig:exp},
we have a more intuitive representation of the consumption of time and space for each instance,
with results of the sample-driven solving procedure in red and results of the isolation-based one in blue.
Here, the above markers reflect the space consumption
while the below line charts reflect the time consumption.
The solid square markers indicate that the observing expressions satisfy the repeated reachability
corresponding to the random intervals,
while the hollow ones refute.
The dotted vertical line indicates the difference in space consumption
between the sample-driven solving procedure and the isolation-based one for that instance.
We can see:
\begin{itemize}
\item When dealing with the single-signal expressions,
	the respective difference in space consumption of the two procedures is not obvious
	as the degree and height vary. 
	The space consumption of both procedures greatly increases
	when solving the constraints containing a CNF of multiple signals.
\item For the same instance,
	the sample-driven solving procedure shows more efficient advantages than the isolation-based one.
	When dealing with more complicated constraints containing a CNF of multiple signals
	instead of a single signal,
	the incremental time consumed by the sample-driven solving procedure
	is much less than that by the isolation-based one.
\end{itemize}
Compared to the isolation-based procedure,
the sample-driven solving procedure saves $83\%$ in time consumption and $9\%$ in space consumption
on average for single-signal constraints,
and saves up to $59\%$ in time consumption and $7\%$ in space consumption
on average for constraints composed of multiple signals.
Apparently,
the sample-driven solving procedure demonstrates great superiority over the isolation-based one 
and is believed to efficient and scalable when encountering complicated situations.

\section{Conclusion}\label{S8}
In this paper,
we have studied the repeated reachability problem over QCTMCs.
First, the decidability was reduced to
the real root isolation of exponential polynomials.
To speed up the procedure, we presented a sample-driven procedure,
which could effectively refine the sample space after each time of sampling,
no matter whether the sample itself was successful or conflicting.
Randomly generated instances had validated the improvement on efficiency.
For future work, we will explore two aspects:
\begin{enumerate}
	\item The proposed method could be applied to verify
	the repeated reachability $\Box^\I\,\Diamond^\J\,\Phi$ and $\omega$-regular properties
	of the general real-time linear system.
	\item The repeated reachability problems in infinite horizon $\Box^\I\,\Diamond\,\Phi$,
	$\Box\,\Diamond^\J\,\Phi$ and $\Box\,\Diamond\,\Phi$
	could be considered for developing the $\mu$-calculus~\cite{BrW18} against QCTMCs.
\end{enumerate}

\section*{Acknowledgments}
The authors thank Yuan Feng for theoretical clarification and Yijia Chen for the suggestion on $\mu$-calculus.

\end{document}